\documentclass{lmcs}
\pdfoutput=1
\usepackage[utf8]{inputenc}

\usepackage{lastpage}
\lmcsdoi{22}{2}{16}
\lmcsheading{}{\pageref{LastPage}}{}{}%
{May~20,~2025}{May~07,~2026}{}

\usepackage{iftex}
\PassOptionsToPackage{[expansion=false]}{microtype}

\ifpdf
  \usepackage{underscore}         
  \usepackage[T1]{fontenc}        
\else
  \usepackage{breakurl}           
\fi

\usepackage{amsmath,amssymb, amsthm, thm-restate}
\usepackage{stmaryrd}
\usepackage{tikz}
\usepackage{url}
\usepackage{adjustbox}

\theoremstyle{definition}
\newtheorem{definition}[thm]{Definition}

\theoremstyle{plain}
\newtheorem{lemma}[thm]{Lemma}

\newtheorem{theorem}[thm]{Theorem}

\renewcommand{\a}{\alpha}

\newcommand{\s}{\sigma}

\newcommand{\e}{\epsilon}

\newcommand{\rr}{0}
\newcommand{\bb}{1}
\renewcommand{\gg}{2} 
\newcommand{\yy}{3} 
\newcommand{\ww}{4} 
\newcommand{\Aa}{\mathcal{A}}
\newcommand{\Ll}{\mathcal{L}}

\newcommand{\Ss}{\mathcal{S}}

\newcommand{\xra}{\xrightarrow}

\newcommand{\incl}{\subseteq}

\newcommand{\prfx}{\sqsubseteq_{\mathsf {pr}}}
\newcommand{\nprfx}{\not\sqsubseteq_{\mathsf {pr}}}
\newcommand{\pprfx}{\sqsubset_{\mathsf {pr}}}
\newcommand{\sfx}{\sqsubseteq_{\mathsf {sf}}}

\newcommand{\NP}{\mathsf{NP}}

\newcommand{\out}{\operatorname{Out}}
\newcommand{\outp}{\overline{\operatorname{Out}}}

\newcommand{\lsfx}[1]{\sqsubseteq_{\scriptscriptstyle
    \mathsf{lsf}}^{\scriptscriptstyle #1}}

\renewcommand{\e}{\varepsilon}
\renewcommand{\epsilon}{\e}

\newcommand{\spath}[3]{\ensuremath{\mathsf{SP}(#1 \rightsquigarrow #2,
#3)}}
\newcommand{\spaths}[2]{\ensuremath{\mathsf{SP}(#1, #2)}}

\begin{document}

\title{Deterministic Suffix-reading Automata}

\author[R. Keerthan]{R Keerthan\lmcsorcid{0009-0003-8278-0951}}[a]
\address{Tata Consultancy Services Innovation Labs \\ Pune, India \& Chennai Mathematical Institute, India}
\email{keerthan@cmi.ac.in}

\author[B. Srivathsan]{B Srivathsan\lmcsorcid{0000-0003-2666-0691}}[b]
\address{Chennai Mathematical Institute, India \& CNRS IRL 2000, ReLaX, Chennai, India}
\email{sri@cmi.ac.in}

\author[R. Venkatesh]{R Venkatesh\lmcsorcid{0009-0007-7747-4457}}[c]

\author[S. Verma]{Sagar Verma\lmcsorcid{0009-0000-9987-4151}}[c]
\address{Tata Consultancy Services Innovation Labs \\
Pune, India}
\email{venkatesh.rv@gmail.com, verma.sagar2@tcs.com}

\begin{abstract}
  We introduce deterministic suffix-reading automata (DSA), a new
  automaton model over finite words. Transitions in a DSA are labeled
  with words. From a state, a DSA triggers an outgoing transition on seeing
  a word \emph{ending} with the transition's label. Therefore,
  rather than moving along an input word letter by letter, a DSA can
  jump along blocks of letters, with each block ending in a suitable
  suffix. This feature allows DSAs to recognize regular languages more
  concisely, compared to DFAs. In this work, we focus on questions
  around finding a ``minimal'' DSA for a regular language. In this context, the number
  of states is not a faithful measure of the size of a DSA, since the
  transition-labels contain strings of arbitrary length. Hence, we
  consider total-size (number of states + number of edges + total
  length of transition-labels) as the size measure of DSAs.

  We start by formally defining the model and providing a DSA-to-DFA
  conversion that allows to compare the expressiveness and
  succinctness of DSA with related automata models.  Our main
  technical contribution is a method to \emph{derive} DSAs from a
  given DFA: a DFA-to-DSA conversion. 
  A surprising
  observation is that the smallest DSA derived from the canonical DFA of
  a regular language $L$ need not be a minimal DSA for $L$. This
  observation leads to a fundamental bottleneck in deriving a minimal
  DSA for a regular language. In fact, we prove that
  given a DFA and a number $k \ge 0$, the problem of deciding if there
  exists an equivalent DSA of total-size $\le k$ is NP-complete. On the other hand, we impose a restriction on the DSA model and show that our derivation procedure can be adapted to produce a minimal automaton within this class of restricted DSAs, starting from the canonical DFA.

\end{abstract}

\maketitle

\section{Introduction}

Deterministic Finite Automata (DFA) are fundamental to many areas in
Computer Science. Apart from being the cornerstone in the study of
regular languages, automata have been applied in several contexts:
such as text processing~\cite{DBLP:journals/tcs/MohriMW09}, 
model-checking~\cite{DBLP:books/daglib/0007403-2}, software
verification~\cite{DBLP:conf/cav/BouajjaniJNT00,
  DBLP:conf/cav/BouajjaniHV04, 989841}, and formal specification
languages~\cite{DBLP:journals/scp/Harel87}. The popularity of automata has also led to the development
of several automata handling libraries~\cite{Awali2.2}.

A central challenge in the application of automata, in almost all of these
contexts, is the large size of the DFA involved. 
Although automata are an intuitive way to describe computations using states and transitions, their descriptions are at a low level, resulting in automata of large size for non-trivial applications. Our objective in this paper is to investigate an automaton model with a more concise representation of the computation.  
Non-determinism is one way that is well-known to give exponential succinctness. However, a
deterministic model is useful in formal specifications and automata
implementations. The literature offers several ways to tackle the problem
of large DFA size, while staying within the realm of determinism. In this
work, we propose a new solution to this problem.

The size of a DFA could be large due to various reasons. One of them is simply the size of the alphabet. For
instance, consider the alphabet of all ASCII characters. Having a
transition for each letter from each state blows up the size of the
automata. \emph{Symbolic automata}~\cite{DBLP:conf/popl/VeanesHLMB12,
  DBLP:conf/cav/DAntoniV17} have been proposed to handle large
alphabets. Letters on the edges are replaced by formulas, which club
together several transitions between a pair of states into one
symbolic transition. If the domain
is the set of natural numbers, a transition
$q \xra{\operatorname{odd(x)}} q'$ with a predicate
$\operatorname{odd(x)}$ is a placeholder for all odd
numbers~\cite{loris-page}. Symbolic automata have been implemented in many
tools and have been widely applied (see \cite{loris-page} for a list
of tools and applications).

In text processing, automata are used as an index for a large set $U$
of strings. Given a text, the goal is to find if it contains one of
the strings from the set. A na\"ive DFA that recognizes the set $U$ is
typically large. \emph{Suffix automata} (different from our model of \emph{suffix-reading} automata) are DFA that accept the set of all
suffixes of words present in $U$. These automata are used as more
succinct indices to represent $U$ and also make the pattern matching
problem more efficient~\cite{DBLP:journals/tcs/MohriMW09}.

Another dimension for reducing the DFA size is to consider
transitions over a block of letters. \emph{Generalized automata} (GA) are
extensions of non-deterministic finite automata (NFAs) that can
contain strings, instead of letters, on transitions. A word $w$ is
accepted if it can be broken down as $w_1 w_2 \dots w_k$ such that
each segment is read by a transition. This model was defined by
Eilenberg~\cite{DBLP:books/lib/Eilenberg74}, and later
Hashiguchi~\cite{DBLP:conf/icalp/Hashiguchi91} proved that for every
regular language $L$ there is a minimal GA in which the edge labels
have length at most a polynomial function in $m$, where $m$ is the size of the syntactic monoid of $L$.

Giammarresi \emph{et al.}~\cite{giammarresi1999deterministic}
consider \emph{Deterministic Generalized Automata} (DGA) and propose an
algorithm to generate a minimal DGA (in terms of the number of states)
in which the edges have labels of length at most the size of the minimal
DFA. The algorithm uses a method to suppress states of a DFA and create longer labels. The key observation is that minimal DGAs can be derived from
the canonical DFA by suppressing states.
The paper also throws open
a more natural question of minimality in the total-size of the
automaton, which includes the number of states, edges and the sum of
label lengths.

A natural extension from strings on transitions to more complex objects is to consider regular
expressions. The resulting automaton model is called \emph{Expression automata}
in~\cite{DBLP:conf/wia/HanW04}. Expression Automata were already considered
in~\cite{DBLP:journals/tc/BrzozowskiM63} to convert automata to
regular expressions. Every DFA can be converted to a two state
expression automaton with a regular expression connecting them. 
A model of \emph{Deterministic Expression automata} (DEA) was proposed
in~\cite{DBLP:conf/wia/HanW04}. To get determinism, every pair of
outgoing transitions from a state is required to have disjoint regular languages, and
moreover, each expression in a transition needs to be prefix-free: for
each $w$ in the language, no proper prefix of $w$ is present in the
language. This restriction makes deterministic expression automata
less expressive than DFAs: it is shown that DEA languages are
prefix-free regular languages. An algorithm to convert a DFA to a DEA,
by repeated state elimination, is proposed
in~\cite{DBLP:conf/wia/HanW04}. The resulting DEA is minimal in the
number of states.  The issue with generic Expression automata is the high expressivity of
the transition condition, that makes states almost irrelevant. On the
other hand, DEA have restrictions on the syntax that make the model less
expressive than DFAs. We consider an intermediate model.

\paragraph*{Our model.} In this work, we introduce \emph{Deterministic
  Suffix-reading Automata (DSAs)}.  We continue to work with strings
on transitions, as in DGA. However, the meaning of transitions is
different. A transition $q \xra{abba} q'$ is enabled if at $q$, a word
$w$ \emph{ending} with $abba$ is seen, and moreover no other
transition out of $q$ is enabled at any prefix of $w$. Intuitively, the
automaton tracks a finite set of pattern strings at each state. It
stays in a state until one of them appears as the \emph{suffix} of the
word read so far, and then makes the appropriate transition.
We start
with a motivating example. Consider a model for out-of-context \texttt{else} statements, in
relation to \texttt{if} and \texttt{endif} statements in a programming
language. Assume a suitable alphabet $\Sigma$ of characters.  Let
$L_{\texttt{else}}$ be the set of all strings over the alphabet where
(1) there are no nested \texttt{if} statements, and (2) there is an
\texttt{else} which is not between an \texttt{if} and an
\texttt{endif}.  
A DFA for this language performs string matching to detect the
\texttt{if}, \texttt{else} and \texttt{endif}. The DSA is shown in
Figure~\ref{fig:if-else}: at $s_0$, it passively reads letters until
it first sees an \texttt{if} or an \texttt{else}. If it is an
\texttt{if}, the automaton transitions to $s_1$. For instance, on a
word \texttt{abf4fgif} the automaton goes to $s_1$, since it ends with
\texttt{if} and there is no \texttt{else} seen so far. Similarly, at
$s_1$ it waits for one of the patterns \texttt{if} or an
\texttt{endif}. If it is the former, it goes to $s_3$ and rejects,
otherwise it moves to $s_0$, and so on. (Note: in case of multiple matches, the longest match is selected. This avoids non-determinism when reading \texttt{endif})

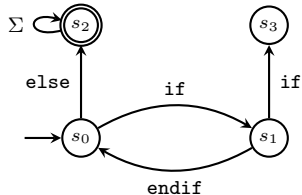
\begin{figure}[t]
  \centering
  \begin{tikzpicture}[state/.style={circle, draw, thick, inner sep =
      2pt}]
    \begin{scope}[every node/.style={state}]
      \node (0) at (0,0) {\tiny $s_0$}; \node (1) at (2.5, 0) {\tiny
        $s_1$}; \node [double] (2) at (0, 1.5) {\tiny $s_2$}; \node
      (3) at (2.5, 1.5) {\tiny $s_3$};
    \end{scope}
    \begin{scope}[->, >=stealth, thick, auto]
      \draw (-0.75, 0) to (0); \draw (0) to node {\scriptsize
        $\mathtt{else}$ } (2); \draw (0) to [bend left=30] node
      {\scriptsize $\mathtt{if}$} (1); \draw (1) to [bend left=30]
      node {\scriptsize $\mathtt{endif}$} (0); \draw (2) to [loop
      left] node {\scriptsize $\Sigma$} (2); \draw (1) to node [right]
      {\scriptsize $\mathtt{if}$} (3);
    \end{scope}
  \end{tikzpicture}
  \caption{DSA for out-of-context \texttt{else}}
  \label{fig:if-else}
\end{figure}

Suffix-reading automata have the ability to wait at a state, reading
long words until a matching pattern is seen. This results
in an arguably more readable specification for languages which are
``pattern-intensive''. This representation is orthogonal to the
approaches considered so far. Symbolic automata club together
transitions between a pair of states, whereas DSA can do this clubbing
across several states and transitions. DGA have this facility of
clubbing across states, but they cannot ignore intermediate letters,
which results in extra states and transitions.

\paragraph*{Overview of results.}
We formally present deterministic suffix-reading automata and its
semantics, quantify its size in comparison to an equivalent DFA, and
study an algorithm to construct a DSA starting from a DFA. This is in
the same spirit as in DGAs, where smaller DGAs are obtained by suppressing
states. For automata models with strings on transitions, the number of
states is not a faithful measure of the size of a DSA. As described
in~\cite{giammarresi1999deterministic} for DGAs, we consider the total-size of
a DSA which includes the number of states, edges, and the sum of label
lengths. The key contributions of this paper are:

\paragraph*{Introduction to the DSA model} We formally present deterministic suffix-reading automata (Section~\ref{sec:new-automaton-model}), and compare its size to DFAs and DGAs (Section~\ref{sec:comparison-with-dfa}). 
We prove that DSAs
accept regular languages, and nothing more.
  Every complete DFA can be seen as a DSA. For the converse, we prove
  that for every DSA of size $k$, there is a DFA with size at most
  $2k \cdot (1 + 2 |\Sigma|)$, where $\Sigma$ is the alphabet
  (Lemma~\ref{lem:tracking-dfa-language-equivalent},
  Theorem~\ref{thm:comparing-dsa-with-dfa-dga}).
  This answers the question of how small DSAs can be in comparison to
  DFAs for a certain language : if $n$ is the size of the minimal DFA
  for a language $L$, minimal DSAs for $L$ cannot be smaller than
  $\frac{n}{2 \cdot (1 + 2
    |\Sigma|)}$. 
  When the alphabet is large, one could expect smaller sized DSAs. We
  describe a family of languages $L_n$, with alphabet size $n$, for
  which the size of the minimal DFA is quadratic in $n$, whereas size of
  the corresponding DSA is a linear function of $n$ (Lemma~\ref{lem:dsa-small}).

\paragraph*{Complexity of minimization} We focus on the minimization problem for DSAs in this paper. Minimization for DFAs can be done efficiently in polynomial-time, thanks to the Myhill-Nerode characterization. For DGAs, \cite{giammarresi1999deterministic} provides a polynomial-time method to find DGAs with smallest number of states, starting from the canonical DFA. For DSAs we prove the following problem to be $\NP$-complete in Section~\ref{sec:complexity}: given a DFA and a number $k$, does there exist a language equivalent DSA with total size at most $k$? This result exhibits an inherent difficulty in minimizing DSAs. 

\paragraph*{A method to derive small DSAs from a given DFA} Similar in spirit to the work on DGAs, we would like to be able to get small DSAs starting from a given DFA. Our main technical constribution is a procedure to derive language equivalent DSAs from a given complete DFA. Of course, the given complete DFA itself can be considered as a DSA. However, our goal is to generate DSAs with as small a total size as possible. 
\begin{itemize}
\item The DFA-to-DSA
  derivation method is presented in Section~\ref{sec:suffix-tracking-sets}. In a nutshell,
  the derivation procedure selects subsets of DFA-states, and adds
  transitions labeled with (some of) the acyclic paths between
  them. The technical challenge lies in identifying sufficient conditions on
  the selected subset of states, so that the derivation procedure
  preserves the language (Theorem~\ref{def:suffix-tracking-set}).

\item We remark that minimal DSAs need not be unique, and make a
  surprising observation: the smallest DSA that we derive from the
  canonical DFA of a language $L$ need not be a minimal DSA. We find this
  surprising because (1) firstly, our derivation procedure is
  surjective: every DSA (satisfying some natural assumptions) can be
  derived from some corresponding DFA (Proposition~\ref{prop:dsa-derivable-from-tracking-dfa}), and in particular, a minimal
  DSA can be derived from some DFA; (2) the observation suggests that
  one may need to start with a bigger DFA in order to derive a minimal
  DSA -- so, starting with a bigger DFA may result in a smaller DSA
  (Section~\ref{sec:minim-some-observ}).

\item Inspired by the above observation, we present a restriction to the definition of DSA, called \emph{strong DSA} (Section~\ref{sec:strong}). We show that minimal strong DSAs can in fact be derived from the canonical DFA, through our derivation procedure. This provides more insights into our derivation procedure and in some sense explains better the capabilities of our procedure. 
\end{itemize}

This introductory work on DSAs opens several threads for future research. We discuss a few of them in the Conclusion (Section~\ref{sec:conclusion}). A shorter version of this work appeared in the conference proceedings~\cite{DBLP:journals/corr/abs-2410-22761}. In the current version, Section~\ref{sec:strong} is completely new and does not appear in the conference proceedings.

\paragraph*{Related work.} The closest to our work
is~\cite{giammarresi1999deterministic} which introduces DGAs, and
gives a procedure to derive DGAs from DFAs. The focus however is on
getting DGAs with as few states as possible. The observations presented in
Section~\ref{sec:minim-some-observ} of our work, also apply for
state-minimality: the same example shows that in order to get a DSA with fewer
states, one may have to start with a bigger DFA.  This is in sharp
contrast to the DGA setting, where the derivation procedure of
\cite{giammarresi1999deterministic} yields a minimal DGA (in the
number of states) when applied on the canonical DFA. The problem of
deriving DGAs with minimal total-size was left open
in~\cite{giammarresi1999deterministic}, and continues to remain so, to
the best of our knowledge.
Expression automata~\cite{DBLP:conf/wia/HanW04} allow regular
expressions as transition labels.  This model was already considered
in~\cite{DBLP:journals/tc/BrzozowskiM63} to convert automata to
regular expressions. Every DFA can be converted to a two state
expression automaton with a regular expression connecting them.  A
model of deterministic Expression automata (DEA) was proposed
in~\cite{DBLP:conf/wia/HanW04} with restrictions that limit the
expressive power.  An algorithm to convert a DFA to a DEA, by repeated
state elimination, is proposed in~\cite{DBLP:conf/wia/HanW04}. The
resulting DEA is minimal in the number of states.  
Minimization of NFAs was studied
in~\cite{DBLP:journals/siamcomp/JiangR93} and shown to be
hard. Succinctness of models with different features, like alternation,
two-wayness, pebbles, and a notion of concurrency, has been studied
in~\cite{DBLP:journals/tcs/GlobermanH96}.

Here are some works that are related to the spirit of finding more readable specifications. \cite{fernau2009algorithms} has used the model of deterministic GA to
develop a learning algorithm for some simple forms of regular
expressions, with
applications in learning DTD specifications for XML code. In this
paper, the author talks about readability of specifications. They
claim that regular expressions (REs) are arguably the best way to
specify regular 
languages. They also claim that the translation algorithms from DFAs
to REs give unreadable REs. In the paper, they consider specific
language classes and generate algorithms to learn simple looking REs.

Another work that talks about readability of specifications
is~\cite{DBLP:conf/icse/ZimmermanLL02}. It considers automata based
specification languages and performs extensive
experiments to determine which choice of syntax gives better
readability. The authors remark that hierarchical specifications are easier,
since not every transition needs to be specified. Once again, we
observe that there is an effort to remove transitions. Another
formalism is proposed in~\cite{DBLP:conf/date/VenkateshSKA14} which
uses a tabular notation where each cell contains patterns, which is
then compiled into a usual automaton for analysis.


\section{Preliminaries}
\label{sec:preliminaries}

We fix a finite alphabet $\Sigma$. Following standard convention, we
write $\Sigma^*$ for the set of all words (including $\e$) over
$\Sigma$, and $\Sigma^+ = \Sigma^* \setminus \{ \e\}$. For
$w \in \Sigma^*$, we write $|w|$ for the length of $w$, with $|\e|$
considered to be $0$.  A word $u$ is a \emph{prefix} of word $w$ if
$w = u v$ for some $v \in \Sigma^*$; it is a \emph{proper-prefix} if
$v \in \Sigma^+$. Observe that $\e$ is a prefix of every word. A set
of words $W$ is said to be a \emph{prefix-free set} if no word in $W$
is a prefix of another word in $W$. A word $u$ is a
\emph{suffix} (resp. \emph{proper-suffix}) of $w$ if $w = vu$ for some
$v \in \Sigma^*$ (resp. $v \in \Sigma^+$). A word $u$ is a \emph{factor} of word $w$ if it is a suffix of a proper-prefix of
$w$, that is $w = v_1 u v_2$ for some $v_1 \in \Sigma^*, v_2 \in \Sigma^+$.

A \emph{Deterministic Finite Automaton (DFA)} $M$ is a tuple
$(Q, \Sigma, q^{init}, \delta, F)$ where $Q$ is a finite set of
states, $q^{init} \in Q$ is the initial state, $F \incl Q$ is a set of
accepting states, and $\delta: Q \times \Sigma \to Q$ is a partial
function describing the transitions. If $\delta$ is complete, the
automaton is said to be a complete DFA. Else, it is called a trim
DFA. The run of DFA $M$ on a word $w = a_1 a_2 \dots a_n$ (where
$a_i \in \Sigma$) is a sequence of transitions
$(q_0, a_1, q_1) (q_1, a_2, q_2) \dots (q_{n-1}, a_n, q_n)$ where $\delta(q_i, a_{i+1}) = q_{i+1}$ for each $0 \le i < n$, and
$q_0 = q^{init}$, the initial state of $M$. The run is accepting if
$q_n \in F$. If the DFA is complete, every word has a unique run. On a
trim DFA, each word either has a unique run, or it has no run. The
language $\Ll(M)$ of DFA $M$, is the set of words for which $M$ has an
accepting run.

We will now recall some useful facts about minimality of DFAs. Here,
by minimality, we mean DFAs with the least number of states. Every
complete DFA $M$ induces an equivalence $\sim_M$ over words:
$u \sim_M v$ if $M$ reaches the same state on reading both $u$ and $v$
from the initial state.  In the case of trim DFAs, this equivalence
can be restricted to set of prefixes of words in $\Ll(M)$. For a
regular language $L$, we have the Nerode equivalence: $u \approx_L v$
if for all $w \in \Sigma^*$, we have $uw \in L$ iff $v w \in L$. By
the well-known Myhill-Nerode theorem (see
\cite{DBLP:books/daglib/0016921} for more details), there is a
canonical DFA $M_L$ with the least number of states for $L$, and
$\sim_{M_L}$ equals the Nerode equivalence $\approx_L$. Furthermore,
every DFA $M$ for $L$ is a \emph{refinement} of $M_L$: $u \sim_M v$
implies $u \sim_{M_L} v$. If two words reach the same state in $M$,
they reach the same state in $M_L$.

  A \emph{Deterministic Generalized Automaton (DGA)}~\cite{giammarresi1999deterministic} $H$ is given by
  $(Q, \Sigma, q^{init}, E, F)$ where $Q, q^{init}, F$ mean the same
  as in DFA, and $E \incl Q \times \Sigma^+ \times Q$ is a finite set
  of edges labeled with words from $\Sigma^+$. For every state $q$,
  the set $\{ \alpha \mid (q, \alpha, q') \in E \}$ is a prefix-free
  set.
A run of DGA $H$ on a word $w$ is a sequence of edges
$(q_0, \alpha_1, q_1) (q_1, \alpha_2, q_2) \dots (q_{n-1}, \a_n, q_n)$
such that $w = \alpha_1 \alpha_2 \dots \a_n$, with $q_0$ being the
initial state. As usual, the run is accepting if $q_n \in F$. Due to
the property of the set of outgoing labels being a prefix-free set,
there is at most one run on every word. The language $\Ll(H)$ is the
set of words with an accepting run. Figure~\ref{fig:dfa-dga-examples}
gives examples of DFAs and corresponding DGAs.

\begin{figure}
  \centering
  \begin{tikzpicture}[state/.style={circle, draw, thick, inner sep =
      2pt}, scale=0.8]
    \begin{scope}[every node/.style={state}]
      \node [double] (0) at (0,0) {\scriptsize $q_0$}; \node (1) at
      (2,0) {\scriptsize $q_1$};
    \end{scope}
    \begin{scope}[->,>=stealth, thick, auto]
      \draw (-0.8, 0) to (0); \draw (0) to [bend left=20] node
      {\scriptsize $a$} (1); \draw (1) to [bend left=20] node
      {\scriptsize $b$} (0);
    \end{scope}
    \node [left] at (-0.7, 0) {\scriptsize $DFA \quad M_1:$};

    \begin{scope}[yshift=-2.5cm]
      \begin{scope}[every node/.style={state}]
        \node [double] (0) at (0,0) {\scriptsize $q_0$};
      \end{scope}
      \begin{scope}[->,>=stealth, thick, auto]
        \draw (-0.8, 0) to (0); \draw (0) to [loop right] node
        {\scriptsize $ab$} (0);
      \end{scope}
      \node [left] at (-0.7, 0) {\scriptsize $DGA \quad H_1:$};
    \end{scope}

    \draw (2.8,-3) to (2.8,1);
    
    \begin{scope}[xshift = 6cm]
      \begin{scope}[every node/.style={state}]
        \node (0) at (0,0) {\scriptsize $q_0$}; \node (1) at (1.8,0)
        {\scriptsize $q_1$}; \node (2) at (3.6,0) {\scriptsize $q_2$};
        \node [double] (3) at (5.4,0) {\scriptsize $q_3$};
      \end{scope}

      \begin{scope}[->,>=stealth, thick, auto]
        \draw (-0.8, 0) to (0); \draw (0) to [loop above] node
        {\scriptsize $b$} (0); \draw (0) to node {\scriptsize $a$}
        (1); \draw (1) to node {\scriptsize $a$} (2); \draw (2) to
        [loop above] node {\scriptsize $a$} (2); \draw (2) to node
        {\scriptsize $b$} (3); \draw (3) to [bend left=20] node
        {\scriptsize $a$} (1); \draw (3) to [bend left=35] node
        {\scriptsize $b$} (0);
       
      \end{scope}
      \node [left] at (-0.7, 0) {\scriptsize $DFA \quad M_2:$};

      \begin{scope}[yshift=-2.5cm]
        \begin{scope}[every node/.style={state}]
          \node (0) at (0,0) {\scriptsize $q_0$}; \node (2) at (3.6,0)
          {\scriptsize $q_2$}; \node [double] (3) at (5.4,0)
          {\scriptsize $q_3$};
        \end{scope}
        \begin{scope}[->,>=stealth, thick, auto]
          \draw (-0.8, 0) to (0); \draw (0) to [loop above] node
          {\scriptsize $b$} (0); \draw (0) to node {\scriptsize $aa$}
          (2); \draw (2) to [loop above] node {\scriptsize $a$} (2);
          \draw (2) to node {\scriptsize $b$} (3); \draw (3) to [bend
          left=20] node {\scriptsize $aa$} (2); \draw (3) to [bend
          left=40] node {\scriptsize $b$} (0);
       
      \end{scope}
      \node [left] at (-0.7, 0) {\scriptsize $DGA \quad H_2:$};
    \end{scope}
  \end{scope}
 
\end{tikzpicture}
\caption{Examples of DFAs and corresponding DGAs, over alphabet
  $\{a, b\}$.}
\label{fig:dfa-dga-examples}
\end{figure}
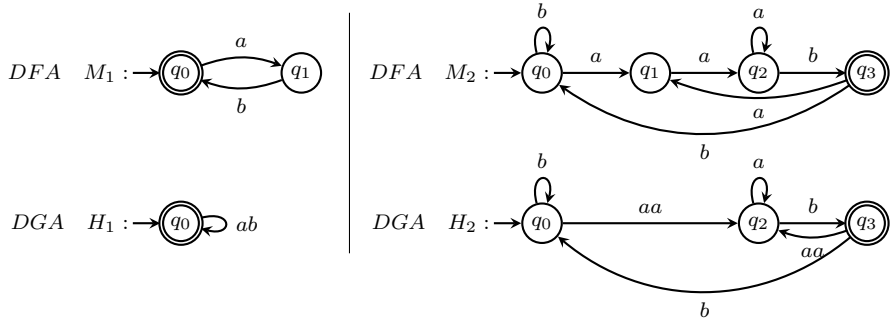

It was
shown in~ \cite{giammarresi1999deterministic} that there is no unique smallest DGA. The paper defines an operation
to suppress states and create longer labels. A state of a DGA is
called \emph{superflous} if it is neither the initial nor final state,
and it has no self-loop. For example, in
Figure~\ref{fig:dfa-dga-examples}, in $M_1$ and $M_2$, state $q_1$ is
superfluous. Such states can be removed, and every pair
$p \xra{\alpha} q$ and $q \xra{\beta} r$ can be replaced with
$p \xra{\alpha \beta} r$.  This operation is extended to a set of
states: given a DGA $H$, a set of states $S$, a DGA $\Ss(H, S)$ is
obtained by suppressing states of $S$, one after the other, in any
arbitrary order. For correctness, there should be no cycle in the
induced subgraph of $H$ restricted to $S$.
The paper proves that minimal DGAs (in number of
states) can be derived by suppressing states, starting from the
canonical DFA.


\section{A new automaton model -- DSA}
\label{sec:new-automaton-model}

We have seen an example of a deterministic suffix automaton in
Figure~\ref{fig:if-else}. A DSA consists of a set of states, and a
finite set of outgoing labels at each state. On an input word $w$, the
DSA finds the earliest prefix which ends with an outgoing label of the
initial state, erases this prefix and goes to the target state of the
transition with the matching label. Now, the DSA processes the rest of
the word from this new state in the same manner. In this section, we
will formally describe the syntax and semantics of DSA.

We start with some more examples. Figure~\ref{fig:example-aab} shows a
DSA for $L_2 = \Sigma^* aab$, the same language as the automata $M_2$
and $H_2$ of Figure~\ref{fig:dfa-dga-examples}. At $q_0$, DSA $\Aa_2$
waits for the first occurrence of $aab$ and as soon as it sees one, it
transitions to $q_3$. Here, it waits for further occurrences of
$aab$. For instance, on the word $abbaabbbaab$, it starts from $q_0$
and reads until $abbaab$ to move to $q_3$. Then, it reads the
remaining $bbaab$ to loop back to $q_3$ and accepts. On a word
$baabaa$, the automaton moves to $q_3$ on $baab$, and continues
reading $aa$, but having nowhere to move, it makes no transition and
rejects the word. Consider another language
$L_3 = \Sigma^*ab\Sigma^*bb$ on the same alphabet $\Sigma$. A similar
machine (as $\Aa_2$) to accept $L_3$ would look like $\Aa_3$ depicted
in Fig.~\ref{fig:example-ab-bb}. For example, on the word $abbbb$, it
would read until $ab$ and move from $q_0$ to $q_1$, read further until
$bb$ and move to $q_2$, then read $b$ and move back to $q_2$ to
accept. We can formally define such machines as automata that
transition on suffixes, or suffix-reading automata.
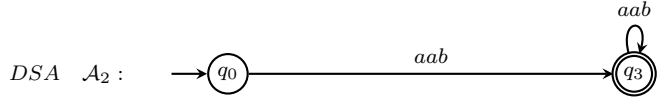
\begin{figure}
  \centering
  \begin{tikzpicture}[state/.style={circle, draw, thick, inner sep =
      2pt}]
    \begin{scope}[every node/.style={state}]
      \node (0) at (0,0) {\scriptsize $q_0$}; \node [double] (1) at
      (5.4,0) {\scriptsize $q_3$};
    \end{scope}
    \begin{scope}[->, thick, >=stealth, auto]
      \draw (-0.75, 0) to (0); \draw (0) to node {\scriptsize $aab$}
      (1); \draw (1) to [loop above] node {\scriptsize $aab$} (1);
      
    \end{scope}
    
    \node [left] at (-1.25, 0) {\scriptsize $DSA \quad \Aa_2:$};
  \end{tikzpicture}
  \caption{DSA $\Aa_2$ accepts $L_2 = \Sigma^*aab$, with
    $\Sigma = \{a, b\}$.}
  \label{fig:example-aab}
\end{figure}

\begin{definition}[DSA]\label{def:suff-reading-aut}
  A \emph{deterministic suffix-reading automaton (DSA)} $\Aa$ is given by a
  tuple $(Q, \Sigma, q^{init}, \Delta, F)$ where $Q$ is a finite set
  of states, $\Sigma$ is a finite alphabet, $q^{init} \in Q$ is the
  initial state, $\Delta \incl Q \times \Sigma^+ \times Q$ is a finite
  set of transitions, $F \incl Q$ is a set of accepting states.  For a
  state $q \in Q$, we define
  $\out(q) := \{ \a \mid (q, \a, q') \in \Delta \text{ for some } q'
  \in Q \}$ for the set of labels present in transitions out of
  $q$. No state has two outgoing transitions with the same label:
  if $(q, \alpha, q') \in \Delta$ and $(q, \alpha, q'') \in \Delta$,
  then $q' = q''$.
  
  The (total) size $|\Aa|$ of DSA $\Aa$ is defined as the sum of
  the number of states, the number of transitions, and the size
  $|\out(q)|$ for each $q \in Q$, where
  $|\out(q)| := \sum_{\a \in \out(q)} |\a|$.

\end{definition}

\begin{figure}[t]
  \centering
  \begin{tikzpicture}[state/.style={circle, draw, thick, inner sep =
      2pt}]
    \begin{scope}[every node/.style={state}]
      \node (0) at (0,0) {\scriptsize $q_0$}; \node (1) at (2,0)
      {\scriptsize $q_1$}; \node [double] (2) at (4,0) {\scriptsize
        $q_2$};
    \end{scope}
    \begin{scope}[->, thick, >=stealth, auto]
      \draw (-0.75, 0) to (0); \draw (0) to node {\scriptsize $ab$ }
      (1); \draw (1) to [bend left=30] node {\scriptsize $bb$} (2);
      \draw (2) to [loop right] node {\scriptsize $b$} (2); \draw (2)
      to [bend left=30] node {\scriptsize $a$} (1);
    \end{scope}
    \node at (-1.75, 0) {\scriptsize $DSA \quad \Aa_3:$};

    \begin{scope}[xshift=8cm]
      \begin{scope}[every node/.style={state}]
        \node (0) at (0,0) {\scriptsize $q_0$}; \node [double] (1) at
        (2,0) {\scriptsize $q_1$};
      \end{scope}
      \begin{scope}[->, thick, >=stealth, auto]
        \draw (-0.75, 0) to (0); \draw (0) to node {\scriptsize $ab$}
        (1); \draw (0) to [loop above] node {\scriptsize $ba$} (0);
      
      \end{scope}

      \node at (-1.75, 0) {\scriptsize $DSA \quad \Aa_4:$};
    \end{scope}
  \end{tikzpicture}
  \caption{$\Aa_3$ accepts $L_3=\Sigma^* ab \Sigma^*bb$ and $\Aa_4$
    accepts $L_4=(b^*ba)^*a^*ab$.}
  \label{fig:example-ab-bb}
\end{figure}
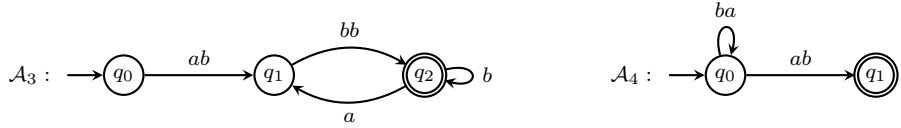

As mentioned earlier, at a state $q$ the automaton waits for a word
that ends with one of its outgoing labels. If more than one label matches,
then the transition with the longest label is taken.  For example, consider the
DSA in Figure~\ref{fig:if-else}. At
state $s_1$ on reading $fghendif$, both the \texttt{if} and
\texttt{endif} transitions match. The longest match is \texttt{endif}
and therefore the DSA moves to $s_0$. This gives a deterministic
behaviour to the DSA. More precisely: at a state $q$, it reads $w$ to
fire $ (q, \a, q') $ if $\a$ is the longest word in $\out(q)$ which is
a suffix of $w$, and no proper prefix of $w$ has any label in
$\out(q)$ as suffix.  We call this a `move' of the DSA. For example,
consider $\Aa_4$ of Figure~\ref{fig:example-ab-bb} as a DSA. Let us
denote $t := (q_0, ab, q_1)$ and $t':= (q_0, ba, q_1)$. We have moves
$(t, ab)$, $(t, aab)$, $(t, aaab)$, and $(t', ba)$, $(t', bba)$,
etc. In order to make a move on $t$, the word should end with $ab$ and
should have neither $ab$ nor $ba$ in any of its proper prefixes.

\begin{definition}\label{def:DSA-moves}
  A \emph{move} of DSA $\Aa$ is a pair $(t, w)$ where
  $t = (q, \a, q') \in \Delta$ is a transition of $\Aa$ and
  $w \in \Sigma^+$ such that
  \begin{itemize}
  \item 
    $\a$ is the longest word in $\out(q)$ which is a suffix of $w$,
    and
  \item 
    no proper prefix of $w$ has a label in $\out(q)$ as
    suffix (no label is a factor of $w$).
  \end{itemize}
  A move $(t, w)$ denotes that at state $q$, transition $t$ gets
  triggered on reading word $w$. We will also write
  $q \xra[\a]{~w~} q'$ for the move $(t, w)$.
\end{definition}

Whether a word is accepted or rejected is determined by a `run' of the
DSA on it. 

\begin{definition}
  A run of $\Aa$ on word $w$, starting from a state $q$, is a sequence
  of moves that consume the word $w$, until a (possibly empty) suffix
  of $w$ remains for which there is no move possible: 
  
  formally, a run
  is a sequence of moves $q_i \xra[\a_i]{~w_i~} q_{i+1}$ for $0\le i\le m$ and $q_m \xra{w_m}$ where $q=q_0$ and $m$ is the index of the last state in the sequence i.e.
  $q = q_0 \xra[\a_0]{~w_0~} q_1 \xra[\a_1]{~w_1~} \cdots
  \xra[\a_{m-1}]{~w_{m-1}~} q_m \xra{w_m}$ such that
  $w=w_0 w_1 \dots w_{m-1} w_m$, and $q_m \xra{w_m}$ denotes that
  there is no move using any outgoing transition from $q_m$ on $w_m$
  or any of its prefixes. 
  
  The run is accepting if $q_m \in F$ and
  $w_m = \e$ (no dangling letters in the end). The language $\Ll(\Aa)$
  of $\Aa$ is the set of all words that have an accepting run starting
  from the initial state $q^{init}$.
\end{definition}

Naturally, the set of words with accepting runs gives the
language of the DSA. Moreover, due to our ``move'' semantics, there is
a unique run for every word.

\section{Comparison with DFA and DGA}
\label{sec:comparison-with-dfa}

Every complete DFA can be seen as an equivalent DSA --- since
$\out(q) = \Sigma$ for every state, the equivalent DSA is forced to
move on each letter, behaving like the DFA that we started off
with. For the DSA-to-DFA direction, we associate a specific DFA to
every DSA, as follows. The idea is to replace transitions of a DSA
with a string-matching-DFA for $\out(q)$ at each
state. Figure~\ref{fig:dsa-to-dfa-eg} gives an example. The
intermediate states correspond to proper prefixes of words in
$\out(q)$.

\begin{definition}[Tracking DFA for a DSA.]\label{def:tracking-dfa}
  For a DSA $\Aa = (Q^\Aa, \Sigma, q_{in}^\Aa, \Delta^\Aa, F^\Aa)$, we
  give a DFA $M_{\Aa}$, called its \emph{tracking DFA}. For
  $q \in Q^\Aa$, let $\outp(q)$ be the set of all prefixes of words in
  $\out(q)$.  States of $M_{\Aa}$ are given by:
  $Q^M = \bigcup_{q \in Q^\Aa} \{ \{ (q, \beta) \mid \beta \in \outp(q)\} \cup
  \{(q,\overline{\epsilon})\} \}$, where $\overline{\epsilon}$ is used as a special symbol to create a copy of each $(q,\epsilon)$. 

  The initial state is $(q^\Aa_{in}, \epsilon)$ and final states are
  $\{ (q, \epsilon) \mid q \in F^\Aa\}$. Transitions are as below: for
  every $q \in Q^\Aa, \beta \in \outp(q), a \in \Sigma$, let $\beta'$ be the
  longest word in $\outp(q)$ s.t $\beta'$ is a suffix of $\beta
  a$. 
  \begin{itemize}
  \item $(q, \beta) \xra{a} (q', \epsilon)$ if $\beta' \in \out(q)$ and  
    $(q,\beta', q') \in \Delta^\Aa$, 
  \item $(q, \beta) \xra{a} (q, \beta')$ if $\beta' \notin \out(q)$ and
    $\beta' \neq \epsilon$,
  \item $(q, \beta) \xra{a} (q,\overline{\epsilon})$ if $\beta' = \epsilon$,
  \item $(q,\overline{\epsilon}) \xra{a} s$, if $(q, \epsilon) \xra{a} s$ according
    to the above (same outgoing transitions).
  \end{itemize}

\end{definition}

\begin{figure}
  \centering
  \begin{tikzpicture}[state/.style={circle, draw, thick, inner sep =
      2pt}]
    \begin{scope}[every node/.style={state}]
      \node (0) at (0,0) {\tiny $q$}; \node [double] (1) at (2,0)
      {\tiny $q'$};
    \end{scope}
    \begin{scope}[->, >=stealth, thick, auto]
      \draw (0) to [bend left=30] node {\scriptsize $abaa$} (1); \draw
      (0) to [bend right=30] node [below] {\scriptsize $baaa$} (1);
    \end{scope}

    \begin{scope}[xshift=3.5cm]
      \begin{scope}[every node/.style={state}]
        \node (0) at (0,0) {\tiny $q$}; \node (a) at (1,1) {\tiny
          $a$}; \node (b) at (1,-1) {\tiny $b$}; \node (ab) at (2, 1)
        {\tiny $ab$}; \node (ba) at (2,-1) {\tiny $ba$}; \node (aba)
        at (3, 1) {\tiny $aba$}; \node (baa) at (3,-1) {\tiny $baa$};
        \node [double] (1) at (4, 0) {\tiny $q'$};
      
      \end{scope}
      \begin{scope}[->,>=stealth, thick, auto]
        \draw (0) to node {\tiny $a$} (a); \draw (0) to node {\tiny
          $b$} (b); \draw (a) to node {\tiny $b$} (ab); \draw (b) to
        node [below] {\tiny $a$} (ba); \draw (ab) to node {\tiny $a$}
        (aba); \draw (ba) to node [below] {\tiny $a$} (baa); \draw
        (aba) to node {\tiny $a$} (1); \draw (baa) to node [below]
        {\tiny $a$} (1); \draw (b) to [loop below] node {\tiny $b$}
        (0); \draw (a) to [loop above] node {\tiny $a$} (a); \draw
        (ba) to node {\tiny $b$} (ab); \draw (baa) to node {\tiny $b$}
        (ab); \draw (ab) to [bend left=10] node [left] {\tiny $b$}
        (b); \draw (aba) to [bend right=60] node [above] {\tiny $b$}
        (ab);
      \end{scope}

    \end{scope}

  \end{tikzpicture}
  \caption{A DSA on the left, and the corresponding DFA for matching
    the strings $abaa$ and $baaa$.}
  \label{fig:dsa-to-dfa-eg}
\end{figure}
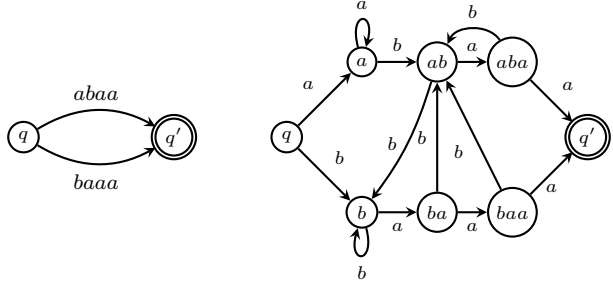

Intuitively, the tracking DFA implements the transition semantics of
DSAs. Starting at $(q, \e)$, the tracking DFA moves along states
marked with $q$ as long as no label of $\out(q)$ is seen as a
suffix. For all such words, the tracking DFA maintains the longest
word among $\outp(q)$ seen as a suffix so far. For instance, in
Figure~\ref{fig:dsa-to-dfa-eg}, at $q$ on reading word $aab$, the DFA
on the right is in state $ab$ (which is the equivalent of $(q, ab)$ in
the tracking DFA definition).

In the rest of the document, we will make use of the following notation:

\begin{definition}[Notation]
  \label{def:notation}\hfill
  
\begin{itemize}
 \item For two words $w_1, w_2$, we write $w_1 \sfx w_2$ if $w_1$
   is a suffix of $w_2$.
 \item For a set of words $U \incl \Sigma^*$ and $\a, \beta \in \Sigma^*$,
 $\a \lsfx{U} \beta$ if $\a$ is the longest word in $U$ which is a
 suffix of $\beta$, i.e. for all $\a' \in U$ if $\a' \sfx \beta$ then
 $|\a'| \leq |\a|$.

\end{itemize}
\end{definition}

\begin{restatable}{lemma}{trackingDFAEquivalent}
  \label{lem:tracking-dfa-language-equivalent}
  For every DSA $\Aa$, the language $\Ll(\Aa)$ equals the language
  $\Ll(M_\Aa)$ of its tracking DFA.
\end{restatable}

\begin{proof}
  The tracking DFA satisfies the following invariants on every state
  $(q, \epsilon)$: (we make use of notations from
  Definition~\ref{def:notation}) 
  \begin{itemize}
  \item Let $w$ be a word that has no $\a \in \out(q)$ as its
    factor. Then, the run of $w$ starting at $(q, \e)$, remains in
    states of the form $(q, \beta)$ and ends in state $(q, \beta_w)$
    where $\beta_w \lsfx{\outp(q)} w$, when $\beta_w \neq \epsilon$; if $\beta_w = \epsilon$, then it ends in state $(q, \overline{\epsilon})$. 
  \item Let $w$ be a word such that $q \xra[\a]{w} q'$ is a move of
    $\Aa$. Then the run of $M_\Aa$ on $w$ starting at $(q, \e)$
    remains in states of the form $(q, \beta)$ and ends in $(q', \e)$.
  \item Let $w$ be a word such that the run of $M_\Aa$ starting at
    $(q, \e)$ remains in states of the form $(q, \beta)$ and ends in a
    state $(q, \beta_w)$. If $\beta_w \neq \overline{\epsilon}$, then $\beta_w \lsfx{\outp(q)} w$. If $\beta_w = \epsilon$, then there is no suffix of $w$ in $\outp(q)$. 
  \item Let $w$ be a word such that the run of $M_\Aa$ starting at
    $(q, \e)$ ends in $(q', \e)$, and all intermediate states are of
    the form $(q, \beta)$. Then there is a move $q \xra[\a]{w} q'$ in
    $\Aa$ where $\a \lsfx{\out(q)} w$.
  \end{itemize}
  All these invariants can be proved by induction on the length of the
  word $w$ and making use of the way the transitions have been defined
  in the tracking DFA. The first two invariants show that every
  accepting run of $\Aa$ has a corresponding accepting run in $M_\Aa$,
  thereby proving $\Ll(\Aa) \incl \Ll(M_\Aa)$. The next two invariants
  show that every accepting run of $M_\Aa$ corresponds to an accepting
  run of $\Aa$, thereby showing $\Ll(M_\Aa) \incl \Ll(\Aa)$.
\end{proof}

Lemma~\ref{lem:tracking-dfa-language-equivalent} and the fact that
every complete DFA is also a DSA, prove that DSAs recognize regular
languages. We will now compare succinctness of DSA wrt DFA and DGA,
starting with a family of languages for which DSAs are concise.

\begin{restatable}{lemma}{dsaSmall}
  \label{lem:dsa-small}
  Let $\Sigma = \{a_1, a_2, \dots, a_n\}$ for some $n \ge 1$. Consider
  the language $L_n = \Sigma^* a_1 a_2 \dots a_n$. There is a DSA for
  this language with size $(4 + 2 |\Sigma|)$.  Any trim DFA for $L_n$ has size at
  least ${|\Sigma|}^2$. (Note: size is counted as the sum of the number of
  states, edges and length of edge labels, in all the automata) 
\end{restatable}

\begin{proof}
  Consider the DSA with states $q_0, q_1$ and transitions
  $q_0 \xra{a_1 \dots a_n} q_1$, $q_1 \xra{a_1 \dots a_n} q_1$, with
  $q_0$ the initial state and $q_1$ the accepting state. This DSA
  accepts $L_n$.

  For any pair $a_1 a_2 \dots a_i$ and $a_1 a_2 \dots a_j$ with
  $i \neq j$, there is a distinguishing suffix:
  $a_1 a_2 \dots a_i \cdot a_{i+1} \dots a_n \in L_n$, but
  $a_1 a_2 \dots a_j \cdot  a_{i+1} \dots a_n \notin
  L_n$. Therefore, the strings $a_1$, $a_1a_2$, $\dots$,
  $a_1 \dots a_n$ go to different states in the canonical DFA. This
  shows there are at least $n$ states in the minimal DFA. Now, from
  $a_1 a_2 \dots a_j$, there is a transition of every letter: clearly,
  there needs to be a transition on $a_{j+1}$; if there is no
  transition on, say $a_\ell \neq a_{j+1}$, then the word
  $a_1 a_2 \dots a_j a_\ell a_1 a_2 \dots a_n$ will be rejected. A
  contradiction. Hence, there are $n$ transitions from each state.
\end{proof}

We make a remark about DGAs for the language family in
Lemma~\ref{lem:dsa-small}. Notice that suppressing superfluous states
of the minimal DFA leads to a strict increase in the label length,
which only strictly increases the (total) size. For instance, suppose
$q_1 \xra{a_2} q_2 \xra{a_3} q_3$ is a sequence of transitions in the
minimal DFA for $L_n$. There are additional transitions
$q_2 \xra{a_1} q_1$ and $q_2 \xra{\Sigma \setminus \{a_1, a_3\}} q_0$
(where $q_0$ is the initial state). Suppressing $q_2$ leads to edges
with labels $a_2a_3$, $a_2a_1$, $a_2 c$ for each
$c \in \Sigma \setminus \{a_1, a_3\}$. Notice that $a_2$ gets copied
$|\Sigma|$ many times. Therefore suppressing states from the minimal
DFA is not helpful in getting a smaller DGA. However, using this
argument, we cannot conclude that the minimal DGA has size at least
$n^2$. Since a characterization of minimality of DGAs in terms of
total size is not known, we leave it at this remark.

We now state the final result of this section, which summarizes the
size comparison between DSAs, DFAs, DGAs. For the comparison to DFAs,
we use the fact that every DSA of size $k$ can be converted to its
tracking DFA, which has atmost $2k$ states. Therefore, size of the
tracking DFA is
bounded by $2k$ (states) $ + 2k \cdot |\Sigma|$ (edges)
$ + 2k \cdot |\Sigma|$ (label length), which comes to
$2k ( 1+ 2|\Sigma|)$.

For a regular language $L$, let
  $n_F^{cmp}(L), n_F^{trim}(L), n_G^{trim}(L), n_S(L)$ denote the size of the
  minimal complete DFA, minimal trim DFA, minimal trim DGA and minimal
  DSA respectively, where size is counted as the sum of the number of
  states, edges and length of edge labels, in all the automata. 

\begin{restatable}{theorem}{comparison}\label{thm:comparing-dsa-with-dfa-dga}
  We have:
  \begin{enumerate}
  \item For all $L$, we have $\dfrac{n_F^{cmp}(L)}{2 (1 + 2|\Sigma|)} \le n_S(L) \le n_F^{cmp}(L)$
  \item 
  There is a language for which $n_S(L)$ is the smallest, and another
    language for which $n_S(L)$ is the largest of the three.
  \end{enumerate}

\end{restatable}

\begin{proof}
  A complete DFA can be seen as a DSA accepting the same
  language. This gives us $n_S(L) \le n_F^{cmp}(L)$. For the inequality
  $\dfrac{n_F^{cmp}(L)}{2 (1 + 2|\Sigma|)} \le n_S(L)$, we show that the
  tracking DFA (Definition~\ref{def:tracking-dfa}) of a DSA has size
  at most $n_S(L) \cdot 2 (1 + 2|\Sigma|)$. The number of states of the
  tracking DFA is atmost $2 n_S(L)$ (recall that $n_S(L)$ denotes the total size of the DSA, which includes the label lengths). For each state there are $\Sigma$
  transitions. Therefore, total size is $2n_S(L)$ (states) + $2n_S(L) \cdot
 |\Sigma|$ (edges) + $2 n_S(L) \cdot |\Sigma|$ (label lengths), which equals
 $2 n_S(L) (1 + 2 |\Sigma|)$.

 For the second part of the proof, Lemma~\ref{lem:dsa-small} gives an
 example where $n_S(L)$ is the smallest. For the other direction,
 consider the language $L_1 = (ab)^*$. A trim DFA for this language has two
 states $q_0, q_1$ (with $q_0$ accepting), and two transitions $q_0 \xra{a} q_1$, and $q_1
 \xra{b} q_0$. Therefore $n_F^{trim}(L_1) \le 2 + 2 + 2 = 6$. A trim DGA
 for this language is $q_0 \xra{ab} q_0$. Therefore $n_G^{trim}(L_1) \le 2
 + 1 + 2 = 5$. We first claim that any DSA for this language needs to
 maintain the following information: (1) the initial state which is
 accepting, as $\e$ is in the 
 language; (2) there is a path from the initial state on
 $ab$ to an accepting state (otherwise $ab$ will not be accepted); (3)
 on reading $aa$ or $b$ from the initial state, the DSA has to make 
 some transitions and go to a sink state, which is a non-accepting state
 --- otherwise,  the words $aab$ or $bab$ will get accepted. This shows that any
 DSA has at least two states, at least three transitions, and ${ab}$, $b$ and ${aa}$ split across some transition labels. This gives either: $n_S(L_1) \ge 2 + 2 + 4 = 8$ for the DSA with two states and $\xra{ab}$, $\xra{b}$ and $\xra{aa}$ as transitions, which is bigger than
 the corresponding trim DFA and DGA, or the min DFA viewed as DSA, with 3 states, and transitions on each of $a$ and $b$ from each of the states giving $n_S \ge3+6+6 = 15$
\end{proof}


\section{Suffix-tracking sets -- obtaining DSA from DFA}
\label{sec:suffix-tracking-sets}

In this section, we are interested in the following task: given a DFA, extract a language equivalent DSA with as small a total size as possible. For instance, given the DFA on the right of Figure~\ref{fig:dsa-to-dfa-eg}, how do we get the DSA in the left of the same figure? In some sense, our method would be to reverse engineer the tracking DFA construction that was introduced in Section~\ref{sec:comparison-with-dfa}, culminating in Definition~\ref{def:tracking-dfa}. 

For DGAs, a
method to derive smaller DGAs by suppressing states was 
recalled in Section~\ref{sec:preliminaries}. The DSA model creates
new challenges. Suppressing states may not always lead to
smaller automata (in total
size). Figure~\ref{fig:minim-challenges-suppr-states} illustrates an
example where suppressing states leads to an exponentially larger
automaton, due to the exponentially many
paths created. But, suppressing states may sometimes
indeed be useful: in Figure~\ref{fig:suppressing-states-reduces-size},
the DFA on the left is performing a string matching to deduce the
pattern $ab$. On seeing $ab$, it accepts. Any extension is
rejected. This is succinctly captured by the DSA on the right. Notice
that the DSA is obtained by suppressing states $q_1$ and $q_3$. So,
suppressing states may sometimes be useful and sometimes not. In~\cite{giammarresi1999deterministic}, the focus was on
getting a DGA with minimal number of states, and hence suppressing
states was always useful.

More importantly, when can we suppress states? DGAs cannot ``ignore''
parts of the word. This in particular leads to the requirement that a
state with a self-loop cannot be suppressed. DSAs have a more
sophisticated transition semantics. Therefore, the procedure to
suppress states is not as simple. This is the subject of this
section. We deviate from the DGA setting in two ways:  we will select a subset of good states from which we can
construct a DSA (essentially, this means the rest of the states are
suppressed); secondly, our starting point will be complete DFA, on
which we make the choice of states (in DGAs, one could start
with any DGA and suppress states). Our procedure can be broken down
into two steps: (1) Start from a complete DFA, select a subset of states and
    build an induced DSA by connecting states using
    acyclic paths between them; (2) Remove some redundant
    transitions.

The plan for this section is as follows.
\begin{description}
\item[Section~\ref{sec:induced-dsa}] We present the core technical concept of suffix-tracking sets~(Definition~\ref{sec:suffix-tracking-sets}): when a set $S$ of states of a given complete DFA $M$ is suffix-tracking, a language equivalent DSA can be induced from $M$. For instance, in Figure~\ref{fig:dsa-to-dfa-eg}, the set $\{q, q'\}$ would be suffix-tracking for the DFA on the right. The DSA induced from $\{q, q'\}$ would contain simple paths from $q$ to $q'$ in $M$. 
\item[Section~\ref{sec:useless}] The induced DSA from suffix-tracking sets can contain useless transitions. For instance, in the same Figure~\ref{fig:dsa-to-dfa-eg}, in addition to simple paths $abaa$ and $baaa$, there are paths $abbaaa$, $babaa$ that move across different branches in the DFA. These are useless transitions, since there are smaller suffixes $baaa$ and $abaa$ respectively that match these patterns. One cannot always simply remove them. Some care needs to be taken while removing transitions. We discuss these observations in this section.
\end{description}
\begin{figure}
  \centering
  
    \begin{tikzpicture}[state/.style={circle, draw, thick, inner sep =
      2pt}, scale=0.8]
    \begin{scope}[every node/.style={state}]
      \node (0) at (0,0) {\tiny $q_0$};
      \node (1) at (2,0) {\tiny $q_1$};
      \node (2) at (4,0)  {\tiny $q_2$};
      \node [double] (3) at (6,0) {\tiny $q_3$};
    \end{scope}
    \begin{scope}[->, >=stealth, thick]
      \draw (-0.8,0) to (0);
      \draw (0) to [bend left] node [above] {\scriptsize $a$} (1); 
      \draw (0) to [bend right] node [below] {\scriptsize $b$} (1);
       \draw (1) to [bend left] node [above] {\scriptsize $a$} (2); 
       \draw (1) to [bend right] node [below] {\scriptsize $b$} (2);
        \draw (2) to [bend left] node [above] {\scriptsize $a$} (3); 
      \draw (2) to [bend right] node [below] {\scriptsize $b$} (3);
    \end{scope}

    \begin{scope}[xshift=8cm]
       \begin{scope}[every node/.style={state}]
      \node (0) at (0,0) {\tiny $q_0$};
      \node [double] (3) at (4,0) {\tiny $q_3$};
    \end{scope}

    \begin{scope}[->, >=stealth, thick]
      \draw (-0.8, 0) to (0);
      \draw (0) to (3);
    \end{scope}

    \node at (2, 0.4) {\scriptsize $aaa, aab, aba, abb$};
    \node at (2, -0.4) {\scriptsize $baa, bab, bba, bbb$};
    
    \end{scope}
  \end{tikzpicture}
  \caption{Suppressing states can add exponentially many labels and
    increase total size.}
  \label{fig:minim-challenges-suppr-states}
\end{figure}
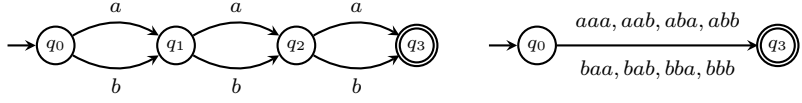

\begin{figure}
  \centering
  \begin{tikzpicture}[state/.style={circle, draw, thick, inner sep =
      2pt}, scale=0.8]
    \begin{scope}[every node/.style={state}]
      \node (0) at (0,0) {\scriptsize $q_0$};
      \node (1) at (1.5,0) {\scriptsize $q_1$};
      \node [double] (2) at (3,0) {\scriptsize $q_2$};
      \node (3) at (4.5,0) {\scriptsize $q_3$};
    \end{scope}
    \begin{scope}[->,>=stealth, thick, auto]
      \draw (-0.8, 0) to (0);
      \draw (0) to [loop above] node {\scriptsize $\Sigma
        \setminus \{a\}$} (0);
      \draw (0) to node {\scriptsize $a$} (1);
      \draw (1) to [bend left = 60] node [below] {\scriptsize $\Sigma
        \setminus \{a, b\}$} (0);
      \draw (1) to [loop above] node {\scriptsize $a$} (1);
      \draw (1) to node {\scriptsize $b$} (2);
      \draw (2) to node {\scriptsize $\Sigma$} (3);
      \draw (3) to [loop above] node {\scriptsize $\Sigma$} (3);
    
    \end{scope}

    \begin{scope}[xshift=8cm]
      \begin{scope}[every node/.style={state}]
        \node (0) at (0,0) {\scriptsize $q_0$};
        \node [double] (2) at (3,0) {\scriptsize $q_2$};
      \end{scope}
      \begin{scope}[->,>=stealth, thick, auto]
        \draw (-0.8, 0) to (0);
        \draw (0) to node {\scriptsize $ab$} (2);
      \end{scope}
    \end{scope}
  \end{tikzpicture}
  \caption{Suppressing states can sometimes reduce total size}
  \label{fig:suppressing-states-reduces-size}
\end{figure}
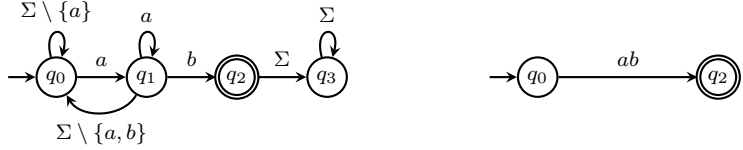

\subsection{Building an induced DSA}
\label{sec:induced-dsa}

We start with an illustrative example. Consider DFA $M$ in
Figure~\ref{fig:induced-eqv}. 
The DSA on the right of the figure shows such an induced DSA obtained
by marking states $\{q_0, q_2\}$ and connecting them using simple
paths. Notice that the language of the induced DSA and the original
DFA are same in this case. Intuitively, all words that end with an $a$
land in $q_1$. Hence, $q_1$ can be seen to ``track'' the suffix $a$.
Now, consider Figure~\ref{fig:induced-not-eqv}. We do the same trick,
by marking states $\{q_0, q_2\}$ and inducing a DSA. Observe that the
DSA does not accept $aba$, and hence is not language equivalent.  When
does a subset of states induce a language equivalent DSA? Roughly,
this is true when the states that are suppressed track ``suitable
suffixes'' (a reverse engineering of the tracking DFA construction of
Definition~\ref{def:tracking-dfa}). As we will see, the suitable
suffixes will be the simple paths from the selected states to the
suppressed states. We begin by formalizing these ideas and then
present sufficient conditions that ensure language equivalence of the
resulting DSA.

\begin{figure}[t]
  \centering
  \begin{tikzpicture}[state/.style={circle, draw, thick, inner sep =
      2pt},scale=0.8]
    \begin{scope}[every node/.style={state}]
      \node (0) at (0,0) {\scriptsize $q_0$}; \node (1) at (2,0)
      {\scriptsize $q_1$}; \node [double] (2) at (4,0) {\scriptsize
        $q_2$};
    \end{scope}
    \begin{scope}[->, thick, >=stealth, auto]
      \draw (-0.75, 0) to (0); \draw (0) to node {\scriptsize $a$ }
      (1); \draw (1) to node {\scriptsize $b$} (2); \draw (0) to [loop
      above] node {\scriptsize $b$} (0); \draw (1) to [loop above]
      node {\scriptsize $a$} (1); \draw (2) to [loop above] node
      {\scriptsize $a,b$} (2);
    \end{scope}
    \node at (-1.25, 0) {\scriptsize $M:$};

    \begin{scope}[xshift=8cm]
      \begin{scope}[every node/.style={state}]
        \node (0) at (0,0) {\scriptsize $q_0$}; \node [double] (1) at
        (2,0) {\scriptsize $q_2$};
      \end{scope}
      \begin{scope}[->, thick, >=stealth, auto]
        \draw (-0.75, 0) to (0); \draw (0) to node {\scriptsize $ab$}
        (1); \draw (1) to [loop above] node {\scriptsize $a,b$} (1);
        \draw (0) to [loop above] node {\scriptsize $b$} (0);
      
      \end{scope}

      \node at (-1.25, 0) {\scriptsize $\Aa_S:$};
    \end{scope}
  \end{tikzpicture}
  \caption{DFA $M$ and an equivalent DSA $\Aa_S$ `induced' with
    $S = \{q_0, q_2\}$.}
  \label{fig:induced-eqv}
\end{figure}
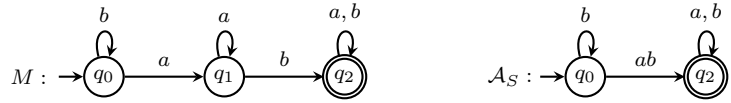

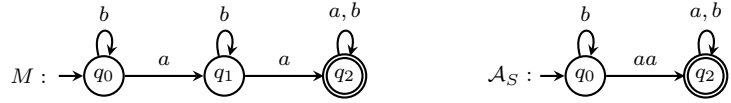
\begin{figure}[t]
  \centering
  \begin{tikzpicture}[state/.style={circle, draw, thick, inner sep =
      2pt}, scale=0.8]
    \begin{scope}[every node/.style={state}]
      \node (0) at (0,0) {\scriptsize $q_0$}; \node (1) at (2,0)
      {\scriptsize $q_1$}; \node [double] (2) at (4,0) {\scriptsize
        $q_2$};
    \end{scope}
    \begin{scope}[->, thick, >=stealth, auto]
      \draw (-0.75, 0) to (0); \draw (0) to node {\scriptsize $a$ }
      (1); \draw (1) to node {\scriptsize $a$} (2); \draw (0) to [loop
      above] node {\scriptsize $b$} (0); \draw (1) to [loop above]
      node {\scriptsize $b$} (1); \draw (2) to [loop above] node
      {\scriptsize $a,b$} (2);
    \end{scope}
    \node at (-1.25, 0) {\scriptsize $M:$};

    \begin{scope}[xshift=8cm]
      \begin{scope}[every node/.style={state}]
        \node (0) at (0,0) {\scriptsize $q_0$}; \node [double] (1) at
        (2,0) {\scriptsize $q_2$};
      \end{scope}
      \begin{scope}[->, thick, >=stealth, auto]
        \draw (-0.75, 0) to (0); \draw (0) to node {\scriptsize $aa$}
        (1); \draw (1) to [loop above] node {\scriptsize $a,b$} (1);
        \draw (0) to [loop above] node {\scriptsize $b$} (0);
      
      \end{scope}

      \node at (-1.25, 0) {\scriptsize $\Aa_S:$};
    \end{scope}
  \end{tikzpicture}
  \caption{DFA $M$ and DSA $\Aa_S$ `induced' with $S = \{q_0,
    q_2\}$. Not equivalent.}
  \label{fig:induced-not-eqv}
\end{figure}

\begin{definition}[Simple words]\label{def:simple-words}
  Consider a complete DFA $M = (Q, \Sigma, q^{init}, \Delta, F)$. Let
  $S \incl Q$ be a subset of states, and $p, q \in Q$. We define
  $\spath{p}{q}{S}$, the \emph{simple words from $p$ to $q$ modulo
    $S$}, as the set of all words $a_1 a_2 \dots a_n \in \Sigma^+$
  such that there is a path:
  $p = p_0 \xra{a_1} p_1 \xra{a_2} \cdots p_{n-1} \xra{a_n} p_n = q$
  in $M$ where
  \begin{itemize}
  \item no intermediate state belongs to $S$:
    $\{ p_1, \dots, p_{n-1}\} \incl Q \setminus S$, and
  \item there is no intermediate cycle: if $p_i = p_j$ for some
    $0 \le i < j \le n$, then $p_i = p_0$ and $p_j = p_n$.
  \end{itemize}
  We write $\spaths{p}{S}$ for $\bigcup_{q \in Q} \spath{p}{q}{S}$,
  the set of all simple words modulo $S$, emanating from $p$.
\end{definition}

For example, in Figure~\ref{fig:induced-eqv}, with $S = \{q_0, q_2\}$,
we have $\spath{q_0}{q_1}{S} = \{a\}$, $\spath{q_0}{q_0}{S} = \{b\}$
and $\spath{q_0}{q_2}{S} = \{ab\}$. These are the same in
Figure~\ref{fig:induced-not-eqv}, except $\spath{q_0}{q_2}{S} = aa$.

Here are some preliminary technical lemmas.
Due to determinism of $M$, we get the following 
property, which underlies several arguments that come later.

\begin{lemma}
  Let $S \incl Q$, and $p, q, r \in Q$ s.t. $q \neq r$. We have
  $\spath{p}{q}{S} \cap \spath{p}{r}{S} = \emptyset$.
\end{lemma}

The next lemma says that every transition either extends a simple path
or is a ``back-edge'' leading to an ancestor in the path.

\begin{lemma}\label{lem:spaths-property}
  Let $S \incl Q$ and $p, q, u \in Q ~ (q \notin S, p \ne q) $ such
  that $q \xra{a} u$ is a transition. For every
  $\s \in \spath{p}{q}{S}$, either $\s a$ belongs to $\spath{p}{u}{S}$, or some
  proper prefix of $\s a$ belongs to $\spath{p}{u}{S}$.
\end{lemma}
\begin{proof}
  Since $\sigma \in \spath{p}{q}{S}$, there is a path
  $p = p_0 \xra{a_1} p_1 \xra{a_2} \cdots p_{n-1} \xra{a_n} p_n = q$,
  with $\s = a_1 a_2 \dots a_n$, satisfying the conditions of
  Definition~\ref{def:simple-words}. If $u \neq p_i$ for
  $0 \le i \le n$, then $\s a \in \spath{p}{u}{S}$ since $q \notin
  S$. If $u = p_i$ for some $0 < i \le n$, then the prefix
  $a_1 a_2 \dots a_{i-1} \in \spath{p}{u}{S}$. If $u = p_0$, then we
  have $\s a \in \spath{p}{u}{S}$.
\end{proof}

Fix a complete DFA $M$ for this section. A DSA can be `induced' from
$M$ using $S$, by fixing states to be $S$ (initial and final states
retained) and transitions to be the simple words modulo $S$ connecting
them i.e. $p \xra{\s} q$ if $\s \in \spath{p}{q}{S}$ (Figure
\ref{fig:induced-eqv}).

\begin{definition}[Induced DSA]\label{def:induced-dsa}
  Given a DFA $M$ and a set $S$ of states in $M$ that contains the
  initial and final states, we define the induced DSA of $M$ (using
  $S$). The states of the induced DSA are given by $S$. The initial
  and final states are the same as in $M$. The transitions are given
  by the simple words modulo $S$ i.e. $p \xra{\s} q$ if
  $\s \in \spath{p}{q}{S}$, for every pair of states $p, q \in S$.
\end{definition}

The induced DSA may not be language-equivalent (Figure
\ref{fig:induced-not-eqv}); to ensure that, we need to check some
conditions. Here is a central definition.

\begin{definition}[Suffix-compatible transitions]\label{def:suffix-compatibility}
  Fix a subset $S \incl Q$. A transition $q \xra{a} u$ is
  suffix-compatible w.r.t. $S$ if either
  $q \in S$ or $u \in S~\textbf{OR}$ for all $p \in S$, and for every
  $\s \in \spath{p}{q}{S}$, there is an $\a \in \spath{p}{u}{S}$ s.t.:
  \begin{itemize}
  \item $\a$ is a suffix of $\s a$, and
  \item moreover, $\a$ is the longest suffix of $\s a$ among words in
    $\spaths{p}{S}$.
  \end{itemize}
\end{definition}

Note that a transition $q \xra{a} u$ is trivially suffix-compatible if
$q \in S$ or $u \in S$. The rest of the condition only needs to be
checked when both of $q,u \notin S$. In
Figure~\ref{fig:induced-not-eqv}, we find the self-loop at $q_1$ to
not be suffix-compatible: we have $S = \{q_0, q_2\}$, and
$\spath{q_0}{q_1}{S} = \{ a \}$, $\spaths{q_0}{S} = \{b, a, ab\}$; the
transition $q_1 \xra{b} q_1$ is not suffix-compatible since there is
no suffix of $ab$ in $\spath{q_0}{q_1}{S}$. Whereas in
Figure~\ref{fig:induced-eqv}, the loop is labeled $a$ instead of
$b$. The transition $q_1 \xra{a} q_1$ is suffix-compatible, since the longest suffix of $aa$ among $\spaths{q_0}{S}$ is $a$ and it is
present in $\spath{q_0}{q_1}{S}$. Let us take the DFA in the right of
Figure~\ref{fig:dsa-to-dfa-eg}, and let $S = \{q, q'\}$. Here are some
of the simple path sets: $\spath{q}{ab}{S} = \{ab, bab, baab\}$,
$\spath{q}{aba}{S} = \{aba, baba, baaba\}$. Consider the transition
$aba \xra{b} ab$. It can be verified that for every
$\s \in \spath{q}{aba}{S}$, the longest suffix of the extension
$\s b$, among simple paths out of $q$, indeed lies in the state $ab$.
In fact, all transitions satisfy suffix-compatibility w.r.t. the
chosen set $S$.

The suffix-compatibility condition is described using simple paths to
states. It requires that every transition take each simple word
reaching its source to the state tracking the longest suffix of its
one-letter extension. This condition on simple paths, transfers to all
words, that circle around the suppressed states. In
Figure~\ref{fig:dsa-to-dfa-eg}, this property can be verified by
considering the word $bbabab$ and its run: $q \xra{b} b \xra{b} b \xra{a} ba \xra{b} ab \xra{a} aba \xra{b} ab$.
At each step, the state reached corresponds to the longest suffix
among the simple words out of $q$.
In the next two lemmas, we prove this claim.

We will use a special
notation: for a state $p \in S$, we write $\out(p, S)$ for
$\bigcup_{r \in S} \spath{p}{r}{S}$; these are the simple words that
start at $p$ and end in some state $r$ of $S$. Notice that these are
the words that appear as transitions in the induced DSA. In
particular, $\out(p)$ in the induced DSA equals $\out(p, S)$. We also remark that $\out(p,S)$ is different from $\spaths{p}{S}$: the latter considers simple words from $p$ to all states in $Q$, whereas the former only considers words from $p$ to $S$.

\begin{restatable}{lemma}{sfxCompatibleDfaToWords}
  \label{lem:sfx-compatibility-dfa-paths-to-words}
  Let $S$ be a set of states such that every transition of $M$ is
  suffix-compatible w.r.t. $S$. Pick $p \in S$, and let
  $w \in \Sigma^+$ be a word with a run
  $p = p_0 \xra{w_1} p_1 \xra{w_2} p_2 \dots p_{n-1} \xra{w_n} p_n$
  such that the intermediate states $p_1, \dots, p_{n-1}$ belong to
  $Q \setminus S$. The state $p_n$ may or may not be in $S$. Then:
  \begin{itemize}
  \item no proper prefix of $w$ has any word from $\out(p,S)$ as
    suffix (no factor of $w$ is in $\out(p,S)$), and
  \item there is $\a \in \spath{p}{p_n}{S}$ such that $\a$ is the
    longest suffix of $w$ among words in $\spaths{p}{S}$.
  \end{itemize}
\end{restatable}

\begin{proof}
  We prove this by induction on the length of the word $w$. When
  $w = a$ for $a \in \Sigma$, we have the run $p \xra{a} q$. The first
  conclusion of the lemma is vacuously true, since there is no
  non-empty proper prefix of $a$. For the second conclusion, note
  that, by definition, we have $a \in \spath{p}{q}{S}$ (in both cases
  when $q \neq p$ and $q = p$). Moreover, clearly $a$ is the longest
  suffix of $a$.

  Now, let $w = w'a$ with a run $p \xra{w'} p' \xra{a} p_n$ such that
  $p'$ and the intermediate states while reading $w'$ belong to
  $Q \setminus S$. Assume the lemma holds for $w'$.  Therefore, the
  longest suffix $\s'$ of $w'$, among $\spaths{p}{S}$, belongs to
  $\spath{p}{p'}{S}$, that is: $\s' \lsfx{\spaths{p}{S}} w'$ and
  $\s' \in \spath{p}{p'}{S}$ (notation as in
  Definition~\ref{def:notation}). We claim that the longest suffix of
  $w'a$ is in fact the longest suffix of $\s' a$.  If not: there
  is $\s'' a$ with $|\s''| > |\s'|$ such that
  $\s'' a \lsfx{\spaths{p}{S}} w' a$. Therefore
  $\s'' \lsfx{\spaths{p}{S}} w'$. This contradicts
  $\s' \lsfx{\spaths{p}{S}} w'$. If $p_n \in Q \setminus S$, the
  longest suffix of $\sigma' a$ lies in $p_n$, by suffix-compatibility
  (Definition~\ref{def:suffix-compatibility}). If $p_n \in S$, we will
  have the exact word $\s' a \in \spath{p}{p_n}{S}$. 
\end{proof}

\begin{restatable}{lemma}{sfxCompatibleWordsToDfa}
  \label{lem:sfx-compatibility-words-to-paths}
  Let $S$ be a set of states such that every transition of $M$ is
  suffix-compatible w.r.t. $S$. Let $p \in S$, and $w \in \Sigma^+$ be
  a word such that no proper prefix of $w$ has a word from
  $\out(p,S)$ as suffix (no factor of $w$ is in $\out(p,S)$). Then:
  \begin{itemize}
  \item The run of $M$ starting from $p$, is of the form
    $p \xra{w_1} p_1 \xra{w_2} p_2 \dots p_{n-1} \xra{w_n} p_n$ where
    $\{p_1, \dots, p_{n-1}\} \incl Q \setminus S$ (notice that we have
    not included $p_n$, which may or may not be in $S$).
  \item the longest suffix of $w$, among $\spaths{p}{S}$ lies in
    $\spath{p}{p_n}{S}$.
  \end{itemize}

\end{restatable}

\begin{proof}
  We prove this by induction on the length of the word $w$. Suppose
  $w = a$, for $a \in \Sigma$. Since $M$ is complete, there is a
  transition $p \xra{a} p_1$. The first item is vacuously
  true. Moreover, if there is $q \in S$ and $\a \in \spath{p}{q}{S}$
  such that $\a$ is a suffix of $a$, then clearly $\a = a$. Hence
  $q = p_1$, due to the determinism of the underlying automaton. If
  there is no such $q$, then it means $p_1\notin S$.

  Suppose $w = w'a$ satisfies the assumptions of the lemma. Assuming
  the lemma holds for $w'$, we have a run $p \xra{w'} p'$ such that
  $p'$ and all intermediate states lie in $Q \setminus S$. Let
  $p' \xra{a} p_n$ be the outgoing transition on $a$ from $p'$. This
  gives a path $p \xra{w'} p \xra{a} p_n$ satisfying the first part of
  the lemma. By the second item of
  Lemma~\ref{lem:sfx-compatibility-dfa-paths-to-words}, the longest
  suffix $\a'a$ of $w'a$ among $\spaths{p}{S}$ belongs to
  $\spath{p}{p_n}{S}$. 
\end{proof}

Suffix-compatibility alone does not suffice to preserve the
language. In Figure~\ref{fig:well-formed-set}, consider
$S = \{0, 2, 4\}$. Every transition is suffix-compatible
w.r.t. $S$. The DSA induced using $S$ is shown in the middle. Notice
that it is not language equivalent, due to the word $aba$ for
instance. The run of $aba$ looks as follows: $0 \xra{ab} 4 \xra{b}
4$. The expected run was $0 \xra{aba} 2$, but that does not happen
since there is a shorter prefix with a matching transition. Even
though, we have suffix-compatibility, we need to ensure that there are
no ``conflicts'' between outgoing patterns. This leads to the next
definition.

\begin{definition}[Well-formed set]\label{def:well-formed-set}
  A set of states $S \incl Q$ is well-formed if there is no
  $p \in S, q \in S$ and $q' \notin S$, with a pair of words
  $\a \in \spath{p}{q}{S}$ (simple word to a state in $S$) and
  $\beta \in \spath{p}{q'}{S}$ (simple word to a state not in $S$)
  such that $\alpha$ is a suffix of $\beta$.
\end{definition}

We observe that the set $S=\{0,2,4\}$ is not well-formed since
$b \in \spath{0}{4}{S}, ab \in \spath{0}{3}{S}$ and $b$ is a suffix of
$ab$. Whereas $S'=\{0,2,3,4\}$ is both suffix-tracking, and
well-formed, and induces an equivalent DSA. On the word $aba$, the run
on the DSA would be $0 \xra{ab} 3 \xra{a} 2$. The first move
$0 \xra{ab} 3$ applies the longest match criterion, and fires the $ab$ transition
since $ab$ is a longer suffix than $b$.  This was not possible before
since $3 \notin S$. It turns out that the two conditions --- suffix-compatibility and
well-formedness --- are sufficient to induce a language equivalent
DSA.

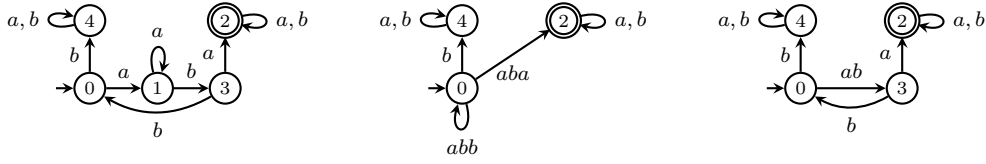
\begin{figure}[t]
  \centering
  \begin{tikzpicture}[state/.style={circle, draw, thick, inner sep =
      2pt}, scale=0.9]
    \begin{scope}[every node/.style={state}]
      \node (0) at (-0.5,0) {\tiny $\rr$}; \node (1) at (0.5, 0)
      {\tiny $\bb$}; \node (2) at (1.5, 0) {\tiny $\yy$}; \node
      [double] (3) at (1.5, 1) {\tiny $\gg$}; \node (4) at (-0.5, 1)
      {\tiny $\ww$};
    \end{scope}
    \begin{scope}[->, thick, >=stealth, auto]
      \draw (-1, 0) to (0); \draw (0) to node {\scriptsize $a$} (1);
      \draw (0) to node {\scriptsize $b$} (4); \draw (1) to [loop
      above] node {\scriptsize $a$} (1); \draw (1) to node
      {\scriptsize $b$} (2); \draw (2) to node {\scriptsize $a$} (3);
      \draw (4) to [loop left] node {\scriptsize $a, b$} (4); \draw
      (3) to [loop right] node {\scriptsize $a, b$} (3); \draw (2) to
      [bend left=30] node {\scriptsize $b$} (0);
    \end{scope}

    \begin{scope}[xshift=5 cm]
      \begin{scope}[every node/.style={state}]
        \node (0) at (0,0) {\tiny $\rr$}; \node [ double] (3) at (1.5,
        1) {\tiny $\gg$}; \node (4) at (0, 1) {\tiny $\ww$};
      \end{scope}
      \begin{scope}[->, thick, >=stealth, auto]
        \draw (-0.5, 0) to (0); \draw (0) to node [below] {\scriptsize
          $aba$} (3); \draw (0) to node {\scriptsize $b$} (4); \draw
        (0) to [loop below] node {\scriptsize $abb$} (0); \draw (4) to
        [loop left] node {\scriptsize $a, b$} (4); \draw (3) to [loop
        right] node {\scriptsize $a, b$} (3);
      
      \end{scope}

    \end{scope}

    \begin{scope}[xshift=10cm]
      \begin{scope}[every node/.style={state}]
        \node (0) at (0,0) {\tiny $\rr$}; \node [double] (3) at (1.5,
        1) {\tiny $\gg$}; \node (2) at (1.5, 0) {\tiny $\yy$}; \node
        (4) at (0, 1) {\tiny $\ww$};
      \end{scope}
      \begin{scope}[->, thick, >=stealth, auto]
        \draw (-0.5, 0) to (0); \draw (0) to node {\scriptsize $ab$}
        (2); \draw (0) to node {\scriptsize $b$} (4); \draw (2) to
        node {\scriptsize $a$} (3); \draw (4) to [loop left] node
        {\scriptsize $a, b$} (4); \draw (3) to [loop right] node
        {\scriptsize $a, b$} (3); \draw (2) to [bend left=30] node
        {\scriptsize $b$} (0);
      \end{scope}

    \end{scope}
  \end{tikzpicture}
  \caption{A DFA, a non-equivalent DSA and an equivalent induced DSA. }
  \label{fig:well-formed-set}
\end{figure}

\begin{definition}[Suffix-tracking sets]\label{def:suffix-tracking-set}
  A set of states $S \incl Q$ is suffix-tracking if it contains the
  initial and accepting states, and
  \begin{enumerate}
  \item every transition of $M$ is suffix-compatible w.r.t. $S$,
  \item and $S$ is well-formed.
  \end{enumerate}
\end{definition}

All these notions lead to the main theorem of this section.

\begin{restatable}{theorem}{suffixTrackingCorrect}
  \label{thm:suffix-tracking-is-correct}
  Let $S$ be a suffix-tracking set of complete DFA $M$, and let
  $\Aa_S$ be the DSA induced using $S$. Then: $\Ll(\Aa_S) = \Ll(M)$
\end{restatable}
\begin{proof}
  Pick $w \in \Ll(M)$. There is an accepting run
  $q_0 \xra{w_1} q_1 \xra{w_2} \dots \xra{w_n} q_n$ of $M$ on $w$. By
  Definition~\ref{def:suffix-tracking-set}, we have $q_0, q_n \in
  S$. Let $1 \le i \le n$ be the smallest index greater than $0$, such
  that $q_i \in S$. Consider the run segment
  $q_0 \xra{w_1} q_1 \xra{w_2} \dots \xra{w_i} q_i$. By
  Lemma~\ref{lem:sfx-compatibility-dfa-paths-to-words}, and by the
  definition of induced DSA~\ref{def:induced-dsa}, no transition of
  $\Aa_S$ out of $q_0$ is triggered until $w_1 \dots w_{i-1}$, and
  then on reading $w_i$, the transition $q_0 \xra{\a} q_i$ is
  triggered, where $\a \in \spath{p}{q}{S}$, and $\a$ is also the
  longest suffix of $w_1 \dots w_i$ among $\spaths{p}{S}$. In
  particular, it is the longest suffix among outgoing labels from
  $q_0$ in $\Aa_S$. This shows there is a move
  $q_0 \xra[\a]{w_1 \dots w_i} q_i$ in $\Aa_S$. Repeat this argument
  on rest of the run
  $q_i \xra{w_{i+1}} q_{i+1} \xra{w_{i+1}} \dots \xra{w_n} q_n$ to
  extend the run of $\Aa_S$ on the rest of the word. This shows
  $w \in \Ll(\Aa_S)$.
  
  Pick $w \in \Ll(\Aa_S)$. There is an accepting run $\rho$ of $\Aa_S$
  starting at the initial state $q_0$. Consider the first move
  $q_0 \xra[\a]{w_1 \dots w_i} q_i$ of $\Aa_S$ on the word. By the
  semantics of a move (Definition~\ref{def:DSA-moves}) and
  Lemma~\ref{lem:sfx-compatibility-words-to-paths}, we obtain a run
  $q_0 \xra{w_1} q_1 \xra{w_2} \dots q_{i-1} \xra{w_i} q_i$ of $M$
  where the intermediate states $q_1, \dots, q_{i-1}$ lie in
  $Q \setminus S$. We apply this argument for each move $\rho$ in the
  accepting run of $\Aa_S$ to get an accepting run of $M$.
\end{proof}

We extend the idea of well-formed set of states to a corresponding notion on DSAs. We will employ this notion when we remove some unnecessary transitions from the induced DSAs that are obtained.

\begin{definition}[Well-formed DSA]\label{def:well-formed-dsa}
  A DSA $\Aa$ is \emph{well-formed} if for every state $q$, no
  outgoing label $\a \in \out(q)$ is a suffix of some proper prefix
  $\beta'$ of another outgoing label $\beta \in \out(q)$.
\end{definition}

\subsection{Removing some redundant transitions.}
\label{sec:useless}

Let us now get back to Figure~\ref{fig:dsa-to-dfa-eg} to see if we can
derive the DSA on the left from the DFA on the right (assuming $q$ is
the initial state). As seen earlier, the set $S = \{q, q'\}$ is
suffix tracking. It is
also well formed since $baaa$ is not a suffix of any prefix of $abaa$
and vice-versa. The DSA $\Aa_S$ induced using $q$ and $q'$ will have
the set of words in $\spath{q}{q'}{S}$ as transitions between $q$ and
$q'$. Both $abaa$ and $baaa$ belong to $\spath{q}{q'}{S}$. However,
there are some additional simple words: for instance, $abbaaa$. Notice
that $baaa$ is a suffix of $abbaaa$, and therefore even if we remove the
transition on $abbaaa$, there will be a move to $q'$ via
$q \xra{baaa} q'$. This tempts us to use only the suffix-minimal words
in the transitions of the induced DSA. This is not always safe, as we
explain below. We show how to carefully remove
``bigger-suffix-transitions''.

Consider the DSA on the left in
Figure~\ref{fig:bigger-suffix-transitions}. If $caba$ is removed, the
moves which were using $caba$ can now be replaced by $ba$ and we still
have the same pair of source and target states. Consider the picture
on the right of the same figure. There is an outgoing edge to a
different state on $aba$. Suppose we remove $caba$. The word $caba$
would then be matched by the longer suffix $aba$ and move to a
different state.
Another kind of redundant transitions are some of the self-loops on
DSAs. In Figure~\ref{fig:induced-eqv}, the self-loop on $b$ at $q_0$
can be removed, without changing the language. This can be generalized
to loops over longer words, under some conditions.

\begin{figure}[t]
  \centering
  \begin{tikzpicture}[state/.style={circle, draw, thick, inner sep =
      2pt}]
    \begin{scope}[every node/.style={state}]
      \node (0) at (0,0) {\tiny $q_0$}; \node (1) at (2,0) {\tiny
        $q_1$};
    \end{scope}
    \begin{scope}[->, >=stealth, thick, auto]
      \draw (0) to node {\tiny $caba, ba$} (1);
    \end{scope}

    \begin{scope}[xshift=4cm]
      \begin{scope}[every node/.style={state}]
        \node (0) at (0,0) {\tiny $q_0$}; \node (1) at (2,0) {\tiny
          $q_1$}; \node (2) at (2,-1) {\tiny $q_2$};
      \end{scope}
      \begin{scope}[->, >=stealth, thick, auto]
        \draw (0) to node {\tiny $caba, ba$} (1); \draw (0) to node
        [below] {\tiny $aba$} (2);
      \end{scope}
    \end{scope}
  \end{tikzpicture}
  \caption{Illustrating bigger-suffix transitions and when they are
    redundant}
  \label{fig:bigger-suffix-transitions}
\end{figure}
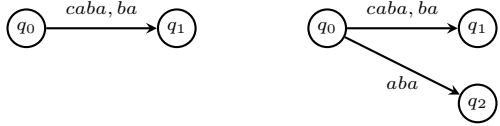

\begin{definition}\label{def:redundant-transitions}
  Let $\Aa$ be a DSA, $q, q'$ be states of $\Aa$ and
  $t:= q \xra{\a} q'$ be a transition.

  We call $t$ a \emph{bigger-suffix-transition} if there exists
  another transition $(q, \beta, q')$ with $\beta$ a suffix of $\alpha$. 

  If there is a transition $t' := q \xra{\gamma} q'' ~ (q'' \neq q')$,
  such that $\beta$ is a suffix of $\gamma$, and $\gamma$ is a suffix
  of $\alpha$, we call $t$ \emph{useful}. A
  bigger-suffix-transition is called \emph{redundant} if it is not
  useful.

  We will say that $t$ is a \emph{redundant self-loop} if $q = q'$,
  $q$ is not an accepting state, and no suffix of $\a$ is a prefix of
  some outgoing label in $\out(q)$.
\end{definition}

In Figure~\ref{fig:bigger-suffix-transitions}, for the automaton on
the left, the transition on $caba$ is redundant. Whereas for the DSA on
the right, $caba$ is a bigger-suffix-transition, but it is
useful. The
self-loop on $q_0$ in Figure~\ref{fig:induced-eqv} is
redundant, but the loop on $q_0$ in Figure~\ref{fig:example-ab-bb}
is useful. Lemmas~\ref{lem:bigger-suffix-transitions-redundant}
and~\ref{lem:removable-self-loop-redundant} 
prove correctness of
removing redundant transitions.

\begin{lemma}
  \label{lem:bigger-suffix-transitions-redundant}
  Let $\Aa$ be a DSA, and let $t:= q \xra{\a} q'$ be a redundant
  bigger-suffix-transition. Let $\Aa'$ be the DSA obtained by removing
  $t$ from $\Aa$. Then, $L(\Aa) = L(\Aa')$.
\end{lemma}
\begin{proof} 
  
 \emph{To show $L(\Aa) \incl L(\Aa')$.} Let $w \in L(\Aa)$ and let
  $q_0 \xra[\a_0]{w_0} q_1 \xra[\a_1]{w_1} \cdots
  \xra[\a_{m-1}]{w_{m-1}} q_m$ be an accepting run. If no
  $(q_i, \a_i, q_{i+1})$ equals $(q, \a, q')$, then the same run is
  present in $S'$, and hence $w \in L(S')$. Suppose
  $(q_j, \a_j, q_{j+1}) = (q, \a, q')$ for some $j$. So, the word
  $w_{j}$ ends with $\a$. As $(q, \a, q')$ is a
  bigger-suffix-transition, there is another $(q, \beta, q')$ such that
  $\beta \sfx \a$. Therefore, the word $w_j$ also ends with $\beta$. Since
  there was no transition matching a proper prefix of $w_j$, the same
  will be true at $\Aa'$ as well, since it has fewer transitions. It
  remains to show that $q_j \xra[\beta]{w_j} q_{j+1}$ is a move. The only
  way this cannot happen is if there is a $q \xra{\gamma} q''$ with
  $\beta \sfx \gamma \sfx \a$. But this is not possible since
  $q \xra{\a} q'$ is a redundant bigger-suffix transition. Therefore,
  every move using $(q, \a, q')$ in $\Aa$ will now be replaced by
  $(q, \beta, q')$ in $\Aa'$. Hence we get an accepting run in $\Aa'$,
  implying $w \in L(\Aa')$.

  \emph{To show $L(\Aa') \incl L(\Aa)$}. Consider $w \in L(\Aa')$ and
  an accepting run
  $q_0 \xra[\a_0]{w_0} q_1 \xra[\a_1]{w_1} \cdots
  \xra[\a_{m-1}]{w_{m-1}} w_m$ in $\Aa'$. Notice that if
  $q \xra[\beta]{w_j} q'$ is a move in $\Aa'$, the same is a move in
  $\Aa$ when $\a \not\sfx w_j$. When $\a \sfx w_j$, then the
  bigger-suffix-transition $q \xra{\a} q'$ will match and the move $q
  \xra[\beta]{w_j} q'$ 
  gets replaced by $q \xra[\a]{w_j} q'$. Hence we will get the same
  run, except that some of the moves using $q \xra{\beta} q'$ may get
  replaced with $q \xra{\a} q'$.
\end{proof}

For the correctness of removing redundant self-loops, we assume that the
DFA that we obtain is well-formed
(Definition~\ref{def:well-formed-dsa}) and has no redundant
bigger-suffix-transitions. The induced DSA that we obtain from
suffix-tracking sets is indeed well-formed. Starting from this induced
DSA, we can first remove all redundant bigger-suffix-transitions, and
then remove the redundant self-loops.

\begin{lemma}
  \label{lem:removable-self-loop-redundant}
  Let $\Aa$ be a well-formed DSA that has no removable
  bigger-suffix-transitions. Let $t := (q, \a, q)$ be a removable
  self-loop. Then the DSA $\Aa'$ obtained by removing $t$ from $\Aa$
  satisfies $\Ll(\Aa) = \Ll(\Aa')$.
\end{lemma}
\begin{proof}
 \emph{To show $\Ll(\Aa) \incl \Ll(\Aa')$}. Let $w \in \Ll(\Aa)$ and
  let $\rho:= q_0 \xra{w_0} q_1 \xra{w_1} \cdots \xra{w_{m-1}} q_m$ be
  an accepting run. Suppose $t$ matches the segment
  $q_j \xra{w_j} q_{j+1}$. Hence $q_j = q_{j+1} = q$.  Observe that as
  $q$ is not accepting, we have $j+1 \neq m$. Therefore there is a
  segment $q_{j+1} \xra{w_{j+1}} q_{j+2}$ in the run. We claim that if
  $t$ is removed, then no transition out of $q$ can match any prefix
  of $w_j w_{j+1}$.

  First we see that no prefix of $w_j$ can be
  matched, including $w_j$ itself: if at all there is a match, it
  should be at $w_j$, and a $\beta$ that is smaller than $\a$. By
  assumption, $\a$ is not a removable
  bigger-suffix-transition. Therefore, there is a transition $q
  \xra{\gamma} q'$, with $\beta \sfx \gamma \sfx \alpha$. This
  contradicts the assumption that $\alpha$ is a removable
  self-loop. Therefore there is no match upto $w_j$.

  Suppose some $(q, \beta, q')$
  matches a prefix $w_j u$ such that $\beta = v u$, that is, $\beta$
  overlaps both $w_j$ and $w_{j+1}$. If $\a \sfx v$, then it violates
  well-formedness of $S$ since it would be a suffix of a proper prefix
  ($v$) of $\beta$. This shows $v \sfx \a$ (since both are suffixes of
  $w_j$) and $v \prfx \beta$, contradicting the assumption that $t$ is
  removable. Therefore, $\beta$ does not overlap $w_j$. But then, if $\beta$
  is a suffix of a proper prefix of $w_{j+1}$, we would not have the
  segment $q_{j+1} \xra{w_{j+1}} q_{j+2}$ in the run
  $\rho$. Therefore, the only possibility is that we have a segment
  $q_j \xra{w_j w_{j+1}} q_{j+2}$. We have fewer occurrences of the
  removable loop $(q, \a, q)$ in the modified run. Repeating this
  argument for every match of $(q, \a, q)$ gives an accepting run of
  $\Aa'$. Hence $w \in L(\Aa')$.

  \emph{To show $L(\Aa') \incl L(\Aa)$.} Let $w \in L(\Aa')$ and
  $\rho' := q_0 \xra{w_0} q_1 \xra{w_1} \cdots \xra{w_{m-1}} q_m$ be
  an accepting run in $\Aa'$. Suppose $q_j \xra{w_j} q_{j+1}$ is
  matched by $(q, \beta, q')$. Let $w_j = v u $ with $\a \sfx v$. Then
  the removable-self-loop $(q, \a, q)$ will match the prefix
  $v$. Suppose $\beta$ overlaps with both $v$ and $u$, that is $\beta
  = \beta' u$. We cannot have $\alpha \sfx \beta'$ due to
  well-formedness of $\Aa$. We cannot have $\beta' \sfx \alpha$ since
  this would mean there is a suffix of $\alpha$ which is a prefix of $\beta$,
  violating the removable-self-loop condition.  Therefore,
  $\beta$ is entirely inside $u$, that is, $\beta \sfx
  u$. Hence in $\Aa$ the run will first start with $q \xra{v}
  q$. Applying the same
  argument, prefixes of the remaining word where 
  $t$ matches will be matched until there is a part of the word where
  $(q, \beta, q')$ matches. This applies to every segment, thereby giving
  us a run in $\Aa$.
\end{proof}

We now get to the core definition of this section, which tells how to
derive a DSA from a DFA, using the methods developed so far.

\begin{definition}[DFA-to-DSA derivation] \label{def:derived-dsa}
  A DSA is said to be \emph{derived from} DFA $M$ using
  $S \subseteq Q$, if it is identical to an induced DSA of $M$
  (using $S$) with all redundant transitions removed.
\end{definition}

By Theorem~\ref{thm:suffix-tracking-is-correct} and Lemma
\ref{lem:bigger-suffix-transitions-redundant}, we get the following
result.

\begin{theorem}
  Every DSA that is derived from a complete DFA using a suffix-tracking set, is language equivalent
  to it.
\end{theorem}


\section{Minimality, some observations and some challenges}
\label{sec:minim-some-observ}

Theorem~\ref{thm:comparing-dsa-with-dfa-dga} shows that we cannot
expect DSAs to be smaller than (trim) DFAs or DGAs in
general. However, Lemma~\ref{lem:dsa-small} and
Figure~\ref{fig:if-else} show that there are cases where DSAs are
smaller and more readable. This motivates us to ask the question of
how we can find a minimal DSA, that is, a DSA of the smallest (total)
size. The first observation is that minimal DSAs need not be unique
--- see Figure~\ref{fig:no-canonical}.

\begin{lemma}
   \label{lem:no-canonical}
   The minimal DSA need not be unique.
 \end{lemma}
 \begin{proof}
  Figure~\ref{fig:no-canonical} illustrates two DSAs, both of size $8$ ($3$ states,
  $2$ edges, total label length $3$), accepting
  the same language. From DSA $\Aa_1$, we can see that the
  language accepted is $b^*a^*abb^*a$. From DSA $\Aa_2$, we can derive
  the expression $b^*aa^*b^*ba$. It is easy to see that both the
  expressions are equivalent.

  It remains to show that there is no DSA of smaller
  size for this language. Let $L = b^*a^*abb^*a = b^*aa^*b^*ba$. Let
  $\Aa$ be a minimal automaton for $L$. Any
  DSA for $L$ has an initial state $q_0$ which is non-accepting (since
  $\e \notin L$, and an accepting state $q_1$ $\neq q_0$). So, in
  particular $\Aa$ has at least two states $q_0, q_1$. The
  smallest word in $L$ is $aba$. Suppose the accepting run of $\Aa$ on
  $aba$ is due to a transition $t:= q_0 \xra{aba} q_1$. Consider the word
  $abba \in L$. Transition $t$ does not match $abba$. Therefore, the
  first transition in the accepting run of $\Aa$ on $abba$ is some
  transition $t' \neq t$. This transition $t'$ will have a label of
  size at least $1$. Hence, $\Aa$ has at least two states, and at
  least two transitions: $t$ which contributes to size $4$ (includes
  $1$ edge and label of length $3$) and a transition $t'$ of size at
  least $2$. Therefore $|\Aa| \ge 2 + 4 + 2 = 8$, which is at least the
  size  of $|\Aa_1|$ and $|\Aa_2|$, and hence any DSA that accepts $aba$ through a single transition has total-size at least $8$.
  Now, suppose the accepting run of $\Aa$ on $aba$ has a
  first transition on a prefix of $aba$, either $a$ or $ab$,
  corresponding to an intermediate transition $q_0 \xra{a} q'$ or $q_0
  \xra{ab} q'$. In the former case, there can be a transition $q'
  \xra{ba} q_1$ to accept $aba$ (if the transition is on just $b$, it only makes the automaton larger since $a$ remains to be read), and in the latter case, a transition
  $q' \xra{a} q_1$. 
  This discussion shows that the two automata $\Aa_2$ and $\Aa_1$,
  are indeed minimal.
\end{proof}

The next  observation
is that a minimal DSAs will be well-formed: if there are two transitions
$q \xra{\a} q_1$ and $q \xra{\beta_1 \a \beta_2} q_2$, then we can remove
the second transition since it will never get fired.

\begin{figure}[t]
  \centering
  \begin{tikzpicture}[state/.style={circle, draw, thick, inner sep =
      2pt}]
    \begin{scope}[every node/.style={state}]
      \node (0) at (0,0) {}; \node (1) at (1, 0) {}; \node [double]
      (2) at (2, 0) {};
    \end{scope}
    \begin{scope}[->,>=stealth, thick, auto]
      \draw (-0.75, 0) to (0); \draw (0) to node {\scriptsize $ab$}
      (1); \draw (1) to node {\scriptsize $a$} (2);
    \end{scope}
    \begin{scope}[xshift=4cm]
      \begin{scope}[every node/.style={state}]
        \node (0) at (0,0) {}; \node (1) at (1, 0) {}; \node [double]
        (2) at (2, 0) {};
      \end{scope}
      \begin{scope}[->,>=stealth, thick, auto]
        \draw (-0.75, 0) to (0); \draw (0) to node {\scriptsize $a$}
        (1); \draw (1) to node {\scriptsize $ba$} (2);
      \end{scope}
    \end{scope}
  \end{tikzpicture}
  \caption{Minimal DSA is not unique}
  \label{fig:no-canonical}
\end{figure}
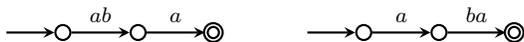

Due to the ``well-formedness'' property in
suffix-tracking sets, the DSAs induced by suffix-tracking sets are
naturally well-formed. Furthermore, since removing redundant transitions preserves
this property, the DSAs that are derived using our DFA-to-DSA
procedure (Definition~\ref{def:derived-dsa}) are well-formed. The next
proposition shows that every DSA that is well-formed and has no
redundant transitions (and in particular, the minimal
DSAs) can be derived from the corresponding tracking DFAs.

\begin{restatable}{proposition}{dsaDerivableFromTracking}
  \label{prop:dsa-derivable-from-tracking-dfa}
  Every well-formed DSA with no redundant transitions can
  be derived from its tracking DFA.
\end{restatable}

\begin{proof}
  We use notation as in Definition~\ref{def:notation}.
  Consider a DSA $\Aa$ that is well-formed and has no redundant
  bigger-suffix-transitions. Let $M_\Aa$ be its tracking DFA as in
  Definition~\ref{def:tracking-dfa}. Let $S = \bigcup_{q\in\Aa}
  {(q,\e)}$. Here is the schema of the proof:
  \begin{align*}
  M_\Aa \xra{~~S~~} \text{ induced DSA } \Aa' \xra{~~\text{remove redundant transitions~~}} \Aa
  \end{align*}
  \begin{enumerate}
  \item We first show that $S$ is suffix-tracking and well-formed.
  \item For each state $q$ of $\Aa$, and for each $\alpha \in \spaths{(q,\epsilon)}{S}$, either $\alpha \in \out(q)$ (a transition out of $q$ in $\Aa$) or $\alpha$ is a redundant bigger-suffix-transition that will be removed in the second step. 
  \end{enumerate}
  All together, this shows that $\Aa$ is derived from $M_\Aa$ using our procedure.

  \paragraph*{$S$ is suffix-tracking.} 
  Pick a state $(q, \beta)$ with $\beta \neq \e$. What is the set
  $\spath{(q, \e)}{(q, \beta)}{S}$? Clearly, $\beta \in \spath{(q,
    \e)}{(q, \beta)}{S}$. Let $\a \in \spath{(q,
    \e)}{(q, \beta)}{S}$, with $\a \neq \beta$. Then there is a path
  $(q, \e) \xra{\a_1} (q, \beta_1) \xra{\a_2} (q, \beta_2) \cdots (q, \beta_{k-1})
  \xra{\a_k} (q, \beta_k) = (q, \beta)$. Let $j \in \{1, \dots, k\}$
  be the largest index such that $\beta_j, \beta_{j+1}, \dots,
  \beta_k$ are all prefixes of $\beta$. Therefore, $\a$ starts by visiting a
  branch different from the prefixes of $\beta$, moves to potentially
  different ``branches'' (prefixes of other words in $\out(q)$, and on
  $\a_1 \dots 
  \a_j$ hits the $\beta$ branch, after which it stays in the same
  branch. By definition, $\beta_j$ is the longest prefix among
  $\outp(q)$ such that $\beta_j \sfx \a_1 \dots \a_j$. We can infer
  the following: (1) $\beta \sfx \a$, and (2) among $\outp(q)$,
  $\beta$ is the longest suffix of $\a$: otherwise, for $\a_1 \dots
  \a_j$, there would be a $\beta' \neq \beta_j$ which is a longer
  suffix, contradicting the definition of the transition $(q,
  \beta_{j-1}) \xra{\a_j} (q, \beta_j)$. 

  These observations are helpful to show that $S$ is
  suffix-tracking. Consider a transition $(q, \beta) \xra{a} (q,
  \beta')$ with $\beta' \neq \e$. We require to show that for every
  $\a \in \spath{(q, \e)}{(q, \beta)}{S}$, the longest suffix of $\a
  a$ among $\spaths{(q, \e)}{S}$ lies in $\spath{(q, \e)}{(q,
    \beta')}{S}$. Pick $\a \in \spath{(q, \e)}{(q, \beta)}{S}$, and
  consider the extension $\a a$. By (1) above, we have $\beta \sfx
  \a$. Since $\beta' \sfx \beta a$, we have $\beta' \sfx \a
  a$. Suppose there is an $\a'' \in \spath{(q, \e)}{(q, \beta'')}{S}$
  such that $|\a''| > |\beta'|$ and $\a'' \sfx \a a$. By definition of
  the transition $(q, \beta) \xra{a} (q,
  \beta')$, $\beta'$ is the longest suffix of $\beta a$, and hence
  $|\beta'| > |\beta''|$. From $|\a''| > |\beta'|$, $\a'' \sfx \a a$
  and $\beta' \sfx \beta a$, we have $\beta' \sfx \a''$. By point (2)
  of the above paragraph,  we have $|\beta''| > |\beta'|$. A
  contradiction. Therefore, there is no such $\a''$. Hence $\beta'$ is
  indeed the longest suffix of $\a a$ for all $\a \in \spath{(q,
    \e)}{(q, \beta)}{S}$.

  Now, consider a transition $(q, \beta) \xra{a} (q,\overline{\epsilon})$. This
  happens when no non-empty word in $\outp(q)$ is a suffix of $\beta
  a$. In that case, there is a transition $(q, \e) \xra{a} (q,\overline{\epsilon})$,
  and $a$ is indeed the longest suffix of $\beta a$ among $\spaths{(q,
    \e)}{S}$. This shows that every transition is suffix-compatible.

  \paragraph*{$S$ is well-formed.} $S$ is well-formed naturally since $\Aa$ is well-formed. Suppose it
  was not. That means there exist states
  $p \in S, q' \notin S, q$, and words $\alpha \in \spath{p}{q}{S}, \beta'
  \in \spath{p}{q'}{S}$, such that $\alpha \sfx \beta'$. Then there is a state $q'' \in S$, and a word $\beta \in \spath{p}{q''}{S}$ such that
  $\beta' \pprfx \beta$, and we have $\a \sfx \beta'$. Since
  $\a,\beta \in \out(p)$, this means $\Aa$ is not well-formed, which is a
  contradiction.

  \paragraph*{Extra strings in the induced DSA are redundant.} In the induced DSA, we will have $\spaths{(q, \e)}{S}$ to have more
  words than $\out(q)$. We need to show that all the other words are
  redundant bigger-suffix-transitions, and hence will be removed by the
  derivation procedure. Consider a transition of the form $(q, \beta)
  \xra{a} (q', \e)$. For every word $\a \in \spath{(q, \e)}{(q,
    \beta)}{S}$, with $|\a| > |\beta|$, we have $\beta$ as the longest
  suffix of $\a$, among $\outp(q)$. Therefore, there is no $\beta'$
  with $|\a| \ge |\beta'| > |\beta|$ with $\beta' \sfx \a$. This is
  sufficient to see that there is no $\beta' a \in \out(q)$ such that
  $\beta a \sfx \beta' a \sfx \a a$. Hence $\a a$ is a redundant
  bigger-suffix-transition in the induced DSA.

  By assumption, we have no redundant bigger-suffix-transitions in
  $\Aa$. Hence, in the derivation procedure, we do not remove any
  transition already present in $\Aa$. This shows that the finally derived DSA is exactly $\Aa$. 
\end{proof}

Proposition~\ref{prop:dsa-derivable-from-tracking-dfa} says that if we
somehow had access to the tracking DFA of a minimal DSA, we will be
able to derive it using our procedure. The challenge however is that
this tracking DFA may not necessarily be the canonical DFA for the
language. In fact, we now show that a smallest DSA that can be derived
from the canonical DFA need not be a minimal DSA.

\begin{figure}
  \begin{tikzpicture}[state/.style={circle, draw, thick, inner sep =
      2pt}]
    \begin{scope}[every node/.style={state}]
      \node (0) at (0,0) {\tiny $q_0$}; \node (1) at (2,1) {\tiny
        $q_1$}; \node (2) at (4,1) {\tiny $q_2$}; \node [double] (3)
      at (4, 0) {\tiny $q_4$}; \node (45) at (3,-1) {\tiny $p$};
    \end{scope}

    \begin{scope}[->,>=stealth, thick, auto]
      \draw (-0.75, 0) to (0); \draw (0) to [loop above] node [left]
      {\tiny $\Sigma \setminus \{a, b\}$} (0); \draw (0) to node
      [below] {\tiny $a$} (1); \draw (1) to [bend right=30] node
      [left] {\tiny $\Sigma \setminus \{a, b\}$} (0); \draw (1) to
      [loop above] node {\tiny $a$} (1); \draw (1) to node {\tiny $b$}
      (2); \draw (2) to [bend left=20] node [below, near start] {\tiny
        $\Sigma \setminus a$} (0); \draw (2) to node {\tiny $a$} (3);
      \draw (0) to node {\tiny $b$} (45);
      
      \draw (45) to node {\tiny $b$} (3);
 
      \draw (45) to 
      [bend left] node {\tiny $\Sigma \setminus \{a,b\}$} (0); \draw
      (45) to [loop below] node {\tiny $a$} (45); \draw (3) to [loop
      right] node {\tiny $\Sigma$} (3);
    \end{scope}

    \begin{scope}[xshift=8cm]
      \begin{scope}[every node/.style={state}]
        \node (0) at (0,0) {\tiny $q_0$}; \node (2) at (4,1) {\tiny
          $q_2$}; \node [double] (3) at (4, 0) {\tiny $q_4$}; \node
        (45) at (3,-1) {\tiny $p$};
      \end{scope}
      
      \begin{scope}[->, >=stealth, thick, auto]
        \draw (-0.75, 0) to (0);
        
        \draw (0) to [bend left] node {\tiny $ab$} (2);
        
        \draw (2) to [bend left=20] node [below, near start] {\tiny
          $\Sigma \setminus a$} (0); \draw (2) to node {\tiny $a$}
        (3); \draw (0) to node {\tiny $b$} (45);
        
        \draw (45) to node {\tiny $b$} (3);
 
        \draw (45) to 
        [bend left] node {\tiny $\Sigma \setminus \{a, b\}$}
        (0); 
        \draw (3) to [loop right] node {\tiny $\Sigma$} (3);
      \end{scope}
    \end{scope}
  \end{tikzpicture}
  \caption{DFA $M^*$ on the left and a derived DSA $\Aa^*_S$ with
    $S = \{q_0, q_2, q_4, p\}$ on the right.}
  \label{fig:dfa-m-star-dsa}
\end{figure}

\begin{figure}
  \centering
  \begin{tikzpicture}[state/.style={circle, draw, thick, inner sep =
      2pt}]
    \begin{scope}[every node/.style={state}]
      \node (0) at (0,0) {\tiny $q_0$}; \node (1) at (2,1) {\tiny
        $q_1$}; \node (2) at (4,1) {\tiny $q_2$}; \node [double] (3)
      at (4, 0) {\tiny $q_4$}; \node [fill=gray!30] (4) at (3,-1)
      {\tiny $p'$}; \node (5) at (3,-2) {\tiny $p$};
    \end{scope}
    \begin{scope}[->,>=stealth, thick, auto]
      \draw (-0.75, 0) to (0); \draw (0) to [loop above] node [left]
      {\tiny $\Sigma \setminus \{a, b\}$} (0); \draw (0) to node
      [below] {\tiny $a$} (1); \draw (1) to [bend right=30] node
      [left] {\tiny $\Sigma \setminus \{a, b\}$} (0); \draw (1) to
      [loop above] node {\tiny $a$} (1); \draw (1) to node {\tiny $b$}
      (2); \draw (2) to [bend left=20] node [below, near start] {\tiny
        $\Sigma \setminus a$} (0); \draw (2) to node {\tiny $a$} (3);
      \draw (0) to node {\tiny $b$} (4); \draw (4) to [bend left] node
      {\tiny $a$} (5); \draw (5) to [bend left] node {\tiny $a$} (4);
      \draw (4) to node {\tiny $b$} (3); \draw (5) to [bend
      right=60]node [right] {\tiny $b$} (3); \draw (4)
      to 
      [bend left] node {\tiny $\Sigma \setminus \{a,b\}$} (0); \draw
      (5) to [bend left=60] node {\tiny $\Sigma \setminus \{a,b\}$}
      (0); \draw (3) to [loop right] node {\tiny $\Sigma$} (3);
    \end{scope}

    \begin{scope}[xshift=8cm]
      \begin{scope}[every node/.style={state}]
        \node (0) at (0,0) {\tiny $q_0$}; \node [double] (3) at (4, 0)
        {\tiny $q_4$}; \node (5) at (3,-2) {\tiny $p$};
      \end{scope}
      \begin{scope}[->, >=stealth, thick, auto]
        \draw (-0.75, 0) to (0); \draw (0) to [bend left] node {\tiny
          $aba$} (3); \draw (0) to [bend right] node {\tiny $bb$} (3);
        \draw (0) to node [right] {\tiny $ba$} (5); \draw (5) to node
        [right] {\tiny $
          b$} (3); \draw (5) to [bend left] node {\tiny
          $\Sigma\setminus \{ a, b\}$} (0); \draw (0) to [loop above]
        node {\tiny $abb$} (0); \draw (3) to [loop right] node {\tiny
          $\Sigma$} (3);
      \end{scope}
    \end{scope}
  \end{tikzpicture}
  \caption{DFA $M^{**}$ on the left and a derived DSA $\Aa^{**}_S$
    with $S = \{q_0, q_2, q_4, p\}$ on the right.}
  \label{fig:m-star-star}
\end{figure}

Figure~\ref{fig:dfa-m-star-dsa} shows a DFA $M^*$. Observe that $M^*$
is minimal: every pair of states has a distinguishing suffix. Let us
now look at DSAs that can be derived from $M^*$. Firstly, any
suffix-tracking set on $M^*$ would contain $q_0, q_4$ (since they are
initial and accepting states). If $p$ is not picked, the transition
$p \xra{a} p$ is not suffix-compatible. Therefore, $p$ should belong
to the selected set. If $p$ is picked, and $q_2$ not picked, then the
set is not well-formed (see Definition~\ref{def:well-formed-set}): the
simple word $b$ from $q_0$ to $p$ is a suffix of the simple word $ab$
to $q_2$. Therefore, any suffix-tracking set should contain the $4$
states $q_0, p, q_2, q_4$. This set $S = \{q_0, p, q_2, q_4\}$ is
indeed suffix-tracking, and the DSA derived using $S$ is shown in the
right of Figure~\ref{fig:dfa-m-star-dsa}. The only other
suffix-tracking set is the set $S'$ of all states. The DSA derived
using $S'$ will have state $q_1$ in addition, and the transitions
$\Sigma \setminus \{a, b\}$. If $\Sigma$ is sufficiently large, this
DSA would have total size bigger than $\Aa^*_S$ (Note: $\Sigma$ can be any alphabet containing $a$ and $b$, so we can increase its size arbitrarily to blow up the DSA). We deduce $\Aa^*_S$
to be the smallest DSA that can be derived from $M^*$.

Figure~\ref{fig:m-star-star} shows DFA $M^{**}$ which is obtained from
$M^*$ by duplicating state $p$ to create a new state $p'$, which is
equivalent to $p$. So $M^{**}$ is language equivalent to $M^*$, but it
is not minimal. Here, if we choose $p$ in a suffix-tracking set, the
simple word to $p$ is $ba$, which is not a suffix of $ab$ (the simple
word to $q_2$). Hence, we are not required to add $q_2$ into the
set. Notice that $S = \{q_0, p, q_4\}$ is indeed a suffix-tracking set
in $M^{**}$. The derived DSA $\Aa^{**}_S$ is shown in the right of the
figure. The ``heavy'' transition on $\Sigma \setminus a$
disappears. There are some extra transitions, like $q_0 \xra{bb} q_4$,
but if $\Sigma$ is large enough, the size of $\Aa^{**}_S$ will be
smaller than $\Aa^*_S$. This shows that starting from a big DFA helps
deriving a smaller DSA, and in particular, the canonical DFA of a
regular language may not derive a minimal DSA for the language.

\section{Strongly Deterministic Suffix-reading Automata (sDSA)}
\label{sec:strong}

As we saw in Section~\ref{sec:minim-some-observ}, the minimal DFA of a
regular language may not derive a minimal DSA for the language (Figures~\ref{fig:dfa-m-star-dsa} and Figure~\ref{fig:m-star-star}). While this is true for the general case, under a restricted definition of DSA we can show the minimal DFA to derive a minimal (restricted) DSA. The plan for this section is as follows:

\begin{description}
\item[Section~\ref{sec:sDSA-syntax}] We first define a class of DSAs with a restricted syntax and call them \emph{strongly deterministic} suffix-reading automata or \emph{strong DSAs} (Definition~\ref{def:str-suff-reading-aut}) and give examples of DSAs that fall under this stronger class.
\item[Section~\ref{sec:sDSA-minimality}]  In the second part, we prove that every minimal strong DSA can be derived from the canonical DFA using the derivation procedure of Section~\ref{sec:suffix-tracking-sets}.
\end{description}

This section is inspired by a question that was left open in the earlier version of this work~\cite{DBLP:journals/corr/abs-2410-22761}: when does the smallest DSA derived from the canonical DFA correspond to a minimal DSA? Although we do not provide an answer to this question, we now understand that when we suitably restrict the DSA syntax, the derivation procedure indeed is able to generate a minimal automaton in the restricted class, starting from the canonical DFA. This provides new insights about the derivation procedure. Moreover, as we will see, many DSAs discussed in this paper are already strong. We also provide a real-life-inspired example of a strong DSA. So, overall, this section sends the message that there are specifications that can be encoded naturally as strong DSAs, and we have a method to generate minimal strong DSAs. 

\subsection{Syntax of strong DSAs}
\label{sec:sDSA-syntax}

We start with a formal description of the syntax of strong DSAs. The examples that follow aim to give an explanation of this definition. In this definition, we say $\alpha'$ is a non-empty prefix of a word $\alpha$ when $\alpha' \neq \epsilon$ and $\alpha' \prfx \alpha$. 

\begin{definition}[sDSA]\label{def:str-suff-reading-aut}
(Strongly Deterministic Suffix-reading Automata). A DSA $\Aa = (Q, \Sigma, q_{init}, \Delta, F)$ is said to be strongly deterministic if for every state $q \in Q$, for every $\a, \beta \in \out (q)$, and for all non-empty prefixes $\a' \prfx \a$ and all non-empty proper prefixes $\beta' \pprfx \beta$, we have $\a' \not \sfx \beta'$ : no prefix of $\a$ is a factor of $\beta$.
\end{definition}

For instance, the DSA in Figure~\ref{fig:if-else} is not an sDSA: at state $s_1$ we have $\out(s_1) = \{ \text{\texttt{if}}, \text{\texttt{endif}}\}$, and the letter \texttt{i} (which is a prefix of \texttt{if}) appears in \texttt{endif} as a suffix  of \texttt{endi}. The DSAs in Figures~\ref{fig:example-aab} and \ref{fig:example-ab-bb} are strong DSAs: it is easy to verify this for $\Aa_2$ and $\Aa_3$; for $\Aa_4$, we provide an explanation. We have $\out(q_0) = \{ab , ba\}$; let $\alpha = ab$ and $\beta= ba$. The non-empty prefixes of $\alpha$ are $a$ and $ab$. The only non-empty proper prefix of $\beta$ is $b$. Notice that neither $a$ nor $ab$ is a suffix of $b$. The same exercise can be repeated with $\alpha = ba$ and $\beta = ab$. We conclude that $\Aa_4$ is an sDSA. Note that any DFA is an sDSA, as the labels of each outgoing transition at a state are distinct letters (and a DFA is a valid DSA). A slightly modified version of the DSA in Figure~\ref{fig:if-else} gives us an example of an sDSA. In Figure~\ref{fig:if-else-no-nest} we have a DSA for out-of-context \texttt{else} statements, where we allow nested \texttt{if}, but a single \texttt{endif} appearing later corresponds to all the open \texttt{if} so far. A context would therefore be the part between the first \texttt{if} and the last \texttt{endif}. If the specification additionally requires to disallow nested \texttt{if}, then it can be modeled by the intersection of the sDSA of Figure~\ref{fig:if-else-no-nest} and another automaton that rejects words with two \texttt{if} with no \texttt{endif} in between.

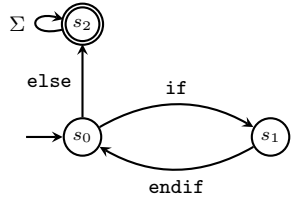
\begin{figure}[t]
  \centering
  \begin{tikzpicture}[state/.style={circle, draw, thick, inner sep =
      2pt}]
    \begin{scope}[every node/.style={state}]
      \node (0) at (0,0) {\tiny $s_0$}; \node (1) at (2.5, 0) {\tiny
        $s_1$}; \node [double] (2) at (0, 1.5) {\tiny $s_2$}; 
    \end{scope}
    \begin{scope}[->, >=stealth, thick, auto]
      \draw (-0.75, 0) to (0); \draw (0) to node {\scriptsize
        $\mathtt{else}$ } (2); \draw (0) to [bend left=30] node
      {\scriptsize $\mathtt{if}$} (1); \draw (1) to [bend left=30]
      node {\scriptsize $\mathtt{endif}$} (0); \draw (2) to [loop
      left] node {\scriptsize $\Sigma$} (2); 
    \end{scope}
  \end{tikzpicture}
  \caption{sDSA for out-of-context \texttt{else}, without nested if statements}
  \label{fig:if-else-no-nest}
\end{figure}

We describe another example of a strong DSA, motivated by a real-life situation. It is a simplified specification of an automotive
application: ``when the alarm is off, and the panic switch is pressed
twice within one clock cycle, go to an error state''. The DFA in
Figure~\ref{fig:automotive-eg} models this. States $q_0$ and $q_1$
represent the alarm being off and on respectively. State $q_3$ is the
error state. The DFA toggles between off and on states on receiving
Signal $s$.  Signal $t$ denotes a ``tick'' marking the separation of
clock cycles, and $p$ denotes the pressing of the panic switch. The
sDSA for this specification is given on the right of
Figure~\ref{fig:automotive-eg}. At state $q_0$, the automaton keeps
receiving signals until either an $s$ or a sequence $pp$ is seen. If
$s$ is seen first, the automaton moves to $q_1$, otherwise it moves to
$q_3$. The move on $q_0 \xra{pp} q_3$ happens on words over
$\{t, p, s\}$ that end with $pp$, and contain neither an $s$ nor a
previous occurrence of $pp$.  This expresses that the panic switch was
pressed twice, within one clock cycle (as no $t$ occurred), while the
alarm was still off. Here, the sDSA provides a smaller, and arguably, a
more readable representation than the DFA.

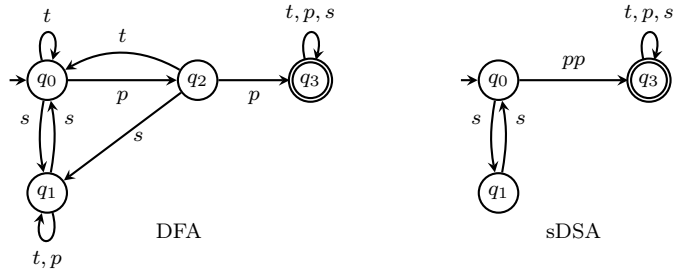
\begin{figure}[t]
  \centering
  \begin{tikzpicture}[state/.style={draw, thick, circle, inner
      sep=2pt}]
    \begin{scope}[every node/.style={state}]
      \node (0) at (0,0) {\scriptsize $q_0$}; \node (1) at (0,-1.5)
      {\scriptsize $q_1$}; \node (2) at (2,0) {\scriptsize $q_2$};
      \node [double] (3) at (3.5,0) {\scriptsize $q_3$};
    \end{scope}
    \begin{scope}[->, thick, >=stealth, auto]
      \draw (-0.5, 0) to (0); \draw (0) to [loop above] node
      {\scriptsize $t$} (0); \draw (0) to [bend right=10] node [left,
      near start] {\scriptsize $s$} (1); \draw (1) to [bend right=10]
      node [right, near end] {\scriptsize $s$} (0); \draw (0) to node
      [below] {\scriptsize $p$} (2); \draw (2) to [bend right=30] node
      [above] {\scriptsize $t$} (0); \draw (2) to node [right]
      {\scriptsize $s$} (1); \draw (2) to node [below] {\scriptsize
        $p$} (3); \draw (3) to [loop above] node {\scriptsize
        $t, p, s$} (3); \draw (1) to [loop below] node {\scriptsize
        $t, p$} (1);
    \end{scope}
    \node at (1.75, -2) {\scriptsize DFA};

    \begin{scope}[xshift=6cm]
      \begin{scope}[every node/.style={state}]
        \node (0) at (0,0) {\scriptsize $q_0$}; \node (1) at (0,-1.5)
        {\scriptsize $q_1$}; \node [double] (3) at (2,0) {\scriptsize
          $q_3$};
      \end{scope}
      \begin{scope}[->, thick, >=stealth, auto]
        \draw (-0.5, 0) to (0); \draw (0) to node {\scriptsize $pp$}
        (3); \draw (0) to [bend right=10] node [left, near start]
        {\scriptsize $s$} (1); \draw (1) to [bend right=10] node
        [right, near end] {\scriptsize $s$} (0); \draw (3) to [loop
        above] node {\scriptsize $t, p, s$} (3);
      \end{scope}
      \node at (1, -2) {\scriptsize sDSA};
    \end{scope}
  \end{tikzpicture}
  \caption{A simplified specification of an automotive application}
  \label{fig:automotive-eg}
\end{figure}

However sDSA may not always provide as succinct a representation as DSA. We see in Figure \ref{fig:dsa-smaller-sdsa-eg}, an example of a DSA on the left and an equivalent sDSA on the right. It is the DSA from Figure \ref{fig:dsa-to-dfa-eg} again, with the corresponding sDSA shown. The DSA itself is not a valid sDSA, since $ba$ is a prefix of $baaa$ and a factor of $abaa$ (also $a$ is prefix of $abaa$ and a factor of $baaa$). Intuitively, to construct an sDSA, we need to break the transitions after $ab$ and $ba$ and introduce new states. This continues to happen on the next transitions from these new states resulting in the sDSA shown.

\begin{figure}
  \centering
  \begin{tikzpicture}[state/.style={circle, draw, thick, inner sep =
      2pt}]
    \begin{scope}[every node/.style={state}]
      \node (0) at (0,0) {\tiny $q$}; \node [double] (1) at (2,0)
      {\tiny $q'$};
    \end{scope}
    \begin{scope}[->, >=stealth, thick, auto]
      \draw (0) to [bend left=30] node {\scriptsize $abaa$} (1); \draw
      (0) to [bend right=30] node [below] {\scriptsize $baaa$} (1);
    \end{scope}

    \begin{scope}[xshift=3.5cm]
      \begin{scope}[every node/.style={state}]
        \node (0) at (0,0) {\tiny $q$}; 
        \node (ab) at (1.5, 1) {\tiny $ab$}; 
        \node (ba) at (1.5,-1) {\tiny $ba$}; 
        \node (aba) at (3, 0) {\tiny $aba$}; 
        \node [double] (1) at (4.5, 0) {\tiny $q'$};
      
      \end{scope}
      \begin{scope}[->,>=stealth, thick, auto]
        \draw (0) to node {\tiny $ab$} (ab); 
        \draw (0) to node [below] {\tiny $ba$} (ba);  
        \draw (ab) to node {\tiny $a$} (aba); 
        \draw (ba) to node [below] {\tiny $a$} (aba); 
        \draw (aba) to node {\tiny $a$} (1); 

        \draw (ba) to node {\tiny $b$} (ab); 
      
        \draw (ab) to [bend right=30] node [left] {\tiny $ba$} (ba); 
        \draw (aba) to [bend right=60] node [above] {\tiny $b$} (ab);
      \end{scope}

    \end{scope}

  \end{tikzpicture}
  \caption{A DSA on the left, and the corresponding (larger) sDSA}
  \label{fig:dsa-smaller-sdsa-eg}
\end{figure}

\subsection{Minimality for strong DSAs}
\label{sec:sDSA-minimality}

 We state some preliminary lemmas, leading up to the main result that any minimal sDSA is derived from a minimal DFA. The next lemma is generic and holds for all DSAs, which are not necessarily strong. We start with some notation. 

 For a state $q$ of a DFA $M$, we define $L^M(q)$ to be the words accepted by $M$ starting from $q$ and call it the \emph{residual language} of the state. Similarly, for a DSA $\Aa$ and a state $q$ of $\Aa$, we can define the residual language $L^\Aa(q)$.
 We say that two states $p, q$ of a DFA/DSA are \emph{equivalent} if their residual languages are equal. In the setting of DFAs, we know that in the canonical (minimal-state) automaton, no two states are equivalent. The next lemma establishes this property even in the DSA setting.

\begin{lemma}\label{lem:minimal-dsa-no-equivalent-states}
No two states of a minimal DSA can be equivalent.
\end{lemma}

\begin{proof}
  Let $\Aa$ be a minimal DSA. Suppose 
  $L^\Aa(p) = L^\Aa(q), p \ne q$. We will
  now construct a DSA $\Aa'$ with smaller size
  than $\Aa$. This will be a contradiction. To get $\Aa'$, we will
  re-orient all transitions going to $q$ to now point towards $p$:
  remove each 
  transition $(r, \a, q)$ and add the transition $(r, \a, p)$; then
  remove $q$ and other unreachable states after this
  transformation. Since $L^\Aa(p) = L^\Aa(q)$, this construction
  preserves the language. For all the states $r$ that remain in $\Aa'$ we
  still have the same $\out(r)$ as in $\Aa$. Finally we need to argue that $|\Aa'| <
  |\Aa|$. This is easy to see since $q$ has been removed from $\Aa$,
  and there are no new additions to $\Aa'$.
\end{proof}

Now we come to properties specific to sDSAs. The key idea is to consider the tracking DFA (Definition~\ref{def:tracking-dfa} and Figure~\ref{fig:dsa-to-dfa-eg}) and investigate equivalence between states in this tracking DFA. 

\begin{lemma}\label{lem:sdsa-equivalent-states}
For any minimal sDSA $\Aa$, its tracking DFA $M_\Aa$ cannot have a `DSA state' equivalent to any other state, i.e. if $(p,\a),(q,\beta)$ are states of $M_\Aa$ and also equivalent, then we have both $\a\ne\e$ and $\beta\ne\e$.
\end{lemma}

\begin{proof}
From Lemma~\ref{lem:minimal-dsa-no-equivalent-states}, we know that no two states of any minimal DSA can be equivalent. So it is not possible to have any distinct $(p,\e),(q,\e) \in M_\Aa$ to be equivalent. Suppose $L^{M_\Aa}(p, \e) = L^{M_\Aa}(q, \beta)$ for some $\beta \ne \e$ (and $q$ could be $p$ as well). We will construct a smaller sDSA $\Aa'$. 

Consider state $q$ of $\Aa$. Due to the presence of state $(q, \beta)$ in $M_\Aa$, there exists some outgoing label of the form $\beta\a$ from $q$ in $\Aa$: that is, $\beta\a \in \out(q)$ for some $\a$ with $|\a| \ge 1$. Let the corresponding transition be $q \xra{\beta\a} r$. To get $\Aa'$ : remove this transition $q\xra{\beta\a}r$ and add the transition $q\xra{\beta}p$. Clearly $\Aa'$ is smaller than $\Aa$, since total size is reduced by $|\a|$. Language accepted is the same. $\Aa'$ is also an sDSA since the only change is replacing $\beta\a$ with $\beta$ in $\out(q)$: any string in $\outp(q)$ that is a factor of $\beta$ would also be a factor of $\beta\a$, and any prefix of $\beta$ is also a prefix of $\beta\a$. This contradicts the assumption that $\Aa$ was a minimal sDSA, thus proving the lemma.
\end{proof}

We remark that the above proof does not work for general DSAs. Consider the transition $q \xra{\beta \alpha} r$ that was removed, and modified to $q \xra{\beta} p$. There could be another transition $q \xra{\beta \alpha'} r'$ in the DSA --- this violates the strong DSA property since $\beta \alpha$ and $\beta \alpha'$ are outgoing labels, and the prefix $\beta$ of the former appears as a suffix of the latter. On seeing the word $\beta \alpha'$ from $q$, the original DSA moves to $r'$. However, in the new DSA, as soon as $\alpha$ is seen the state changes to $p$. This can modify the language. Such a situation does not happen in an sDSA due to the restriction on the outgoing labels. We now prove the main result of this section.

\begin{theorem}
Every minimal sDSA for a language $L$ can be derived from the canonical DFA for $L$.
\end{theorem}

\begin{proof}

  Let $\Aa = (Q^\Aa, \Sigma, \Delta^\Aa, F^\Aa)$ be an arbitrary minimal sDSA.
Any minimal sDSA is well-formed and has no redundant transitions. Hence $\Aa$ can be derived from its tracking DFA $M_\Aa$, thanks to Proposition~\ref{prop:dsa-derivable-from-tracking-dfa}. If the tracking DFA $M_\Aa$ is canonical, we are done. Otherwise, there are distinct states of $M_\Aa$ that are equivalent.  From Lemma~\ref{lem:sdsa-equivalent-states}, we know that a state $(q, \epsilon)$ of $M_\Aa$ cannot be equivalent to any other state. Therefore,  two equivalent states are of the form $(p, \alpha)$ and $(q, \beta)$, with $\alpha$ and $\beta$ both non-empty. We will merge such equivalent states together to build the minimal automaton and make use of this specific structure to show that $\Aa$ can be derived from the minimized DFA. For the proof, we will show the following steps.
\begin{enumerate}
\item Suppose a state $(q, \alpha)$ is equivalent to $(q, \beta)$ (with the same $q$), then $\alpha \not \prfx \beta$ (i.e. the two states cannot be `tracking' the same $\Aa$-transition out of $q$).
\item Build a DFA $M'_\Aa$ by quotienting equivalent states of $M_\Aa$, with $[q]$ representing the equivalence class of $q$. The resulting automaton $M'_\Aa$ is the canonical DFA. We use (1) to show that every simple path from a state $(q, \epsilon)$ to $(q, \beta)$ is preserved in the quotiented automaton $M'_\Aa$, from $[(q, \epsilon)]$ to $[(q, \beta)]$.
\item From (2), we show that the set of states $S' := \{ [(q, \epsilon)] \mid q \in Q^\Aa\}$ forms a suffix-tracking set. The DSA derived using $S'$ turns out to be $\Aa$, proving that $\Aa$ can be derived from the canonical DFA.
\end{enumerate}

\paragraph*{Step 1.} If two states $(q,\a),(q,\beta) \in M_\Aa$ are equivalent, then $\a\nprfx\beta$ (i.e. the two states are not `tracking' the same $\Aa$-transition). In other words, two states in the tracking DFA that lie on the same path corresponding to a given DSA transition, cannot be equivalent. Suppose $\a\prfx\beta$. We know that $\beta\in\outp(q)$, so there is a string $\gamma$ such that $\beta\gamma\in\out(q)$ i.e. reading $\gamma$ from $(q,\beta)$ leads to a `DSA state', say $(q',\e)$. Let $s$ be the state reached on reading $\gamma$ from $(q, \alpha)$. Since $(q, \alpha)$ and $(q, \beta)$ are equivalent, the state $s$ will be equivalent to $(q', \e)$. From Lemma~\ref{lem:sdsa-equivalent-states}, this means $s = (q', \e)$. This gives transition labels $\a\gamma$ (induced by the simple path just discussed) and $\beta\gamma$ from $q$ in $\Aa$, with $\a\prfx\beta$, contradicting the fact that $\Aa$ was an sDSA. Hence we must have $\a\nprfx\beta$.

\paragraph*{Step 2.} For two states $s, s'$ of $M_\Aa$, define $s \equiv_{M_\Aa} s'$ if $L^{M_\Aa}(s) = L^{M_\Aa}(s')$. We denote the equivalence class of $s$ by $[s]$. From Lemma~\ref{lem:sdsa-equivalent-states}, $[(q, \epsilon)] = \{ (q, \epsilon)\}$. Consider a quotient DFA $M'_\Aa$ based on this equivalence. States are the equivalence classes of $\equiv_{M_\Aa}, \{ [s] \mid s \in M_\Aa \}$. 
The initial state is $[(q^\Aa_{in}, \epsilon)]$ and final states are $\{ [(q, \epsilon)] \mid q \in F^\Aa\}$. Transitions are given by $\{ [s] \xra{a} [s'] \mid s \xra{a} s' \in M_\Aa \}$. This is deterministic because whenever we have $[s_1] = [s'_1]$ and $s_1 \xra{a} s_2, s'_1 \xra{a} s'_2 \in M_\Aa$, we also have $[s_2]=[s'_2]$. 
Note that $M_\Aa'$ is minimal, by definition. 

Consider a transition $q \xra{\alpha} q'$ of the sDSA $\Aa$, with $\alpha = a_1 a_2 \dots a_n$. This transitions gives states $(q, \epsilon)$, $(q, a_1)$, $(q, a_1 a_2)$, $\dots$, $(q, a_1 \dots a_{n-1})$, $(q', \epsilon)$ in the tracking DFA $M_\Aa$. From (1), we have $[(q,a_1)],[(q,a_1a_2)],\dots,[(q,a_1a_2\dots a_{n-1})]$ to be distinct states in $M'_\Aa$. equivalent. Hence $a_1 a_2 \dots a_i$ is a simple path from $[(q, \epsilon)]$ to $[(q, a_1 \dots a_i)]$ that does not visit any state of the form $[(p, \epsilon)]$ in between.

\paragraph*{Step 3.} Let $S' := \bigcup_{q\in\Aa}\{[(q,\e)]\}$. We will show that $S'$ is suffix-tracking. Consider a simple path $\s$ from $[(q, \e)]$ to $[(q, \beta)]$ (note that $\s$ is also a simple path from $(q, \e)$ to $(q, \beta)$ in $M_\Aa$). 
In the tracking DFA $M_\Aa$, every transition moves to a state tracking the longest possible suffix: $(q, \beta) \xra{a} (q, \beta')$ in $M_\Aa$ means $\beta'$ is the longest word in $\outp(q)$ s.t $\beta' \sfx \s a$. From (2), $\beta'$ is a simple path from $[(q, \epsilon)]$ to $[(q, \beta')]$, in $M'_\Aa$. Moreover, we have $[(q, \beta)] \xra{a} [(q, \beta')]$ by definition. 

We claim 
that the longest suffix of $\s a$ among simple paths from $[(q,\e)]$, is $\beta'$, which goes from $[(q,\e)]$ to $[(q,\beta')]$. This would show suffix-compatibility of the transition. Suppose instead, the longest suffix (say $\s'$) went from $[(q,\e)]$ to $[(q,\a')]$ in $M'_\Aa$; we have $\s'$ to also be a simple path from $(q, \e)$ to $(q, \a')$ in $M_\Aa$, and $\alpha'$ to be longer than $\beta'$. This is a contradiction. Hence $\beta'$ is the longest simple-path suffix of $\s a$ in $M'_\Aa$, making  $[(q, \beta)] \xra{a} [(q, \beta')]$ suffix-compatible.

The set $S'$ is also well-formed since $\Aa$ is well-formed. Suppose it was not. That means there exist states $p \in S', q' \notin S'$, and words $\alpha \in \spath{p}{q}{S'}, \beta' \in \spath{p}{q'}{S'}$, such that $\alpha \sfx \beta'$. Then there is a state $q'' \in S'$, and a word $\beta \in \spath{p}{q''}{S'}$ such that $\beta' \pprfx \beta$,  and we have $\a \sfx \beta'$. Since $\a,\beta \in \out(p)$ are paths corresponding to $\Aa$-transitions, this means $\Aa$ is not an sDSA, which is a contradiction. Hence we have $S'$ to be suffix-tracking. In the induced DSA $\Aa_{S'}$, every simple path from $[(q, \epsilon)]$ to $[(q', \epsilon)]$ will be present. Hence every transition $q \xra{\alpha} q'$ in $\Aa$ is present in $\Aa_{S'}$. Any transition out of $[(q, \epsilon)]$ which is not in $\out(q)$ in $\Aa$, will be a redundant bigger-suffix-transition, analogous to the argument in last part of the proof of Proposition~\ref{prop:dsa-derivable-from-tracking-dfa}.
 Thus $\Aa$ can be derived from $M'_\Aa$, which is a minimal DFA.
\end{proof}

\section{Complexity of minimization}
\label{sec:complexity}

The goal of this section is to prove the following theorem.

\begin{theorem}
  \label{thm:np-complete}
  Given a DFA $M$ and positive integer $k$, deciding whether there
  exists a DSA of total size $\le k$ language equivalent to $M$ is
  NP-complete.
\end{theorem}

If $k$ is bigger than the size of the DFA $M$, then the answer is
trivial. Therefore, let us assume that $k$ is smaller than the DFA
size. For the $\NP$ upper bound, we guess a DSA of total size $k$,
compute its tracking DFA in time $\mathcal{O}(k \cdot |\Sigma|)$ and
check for its language equivalence with the given DFA $M$. This can be
done in polynomial-time by minimizing both the DFA and checking for
isomorphism.

The rest of the section is devoted to proving the lower bound.  We
provide a reduction from the minimum vertex cover problem which is a
well-known $\NP$-complete problem~\cite{DBLP:conf/coco/Karp72}. A
vertex cover of an undirected graph $G = (V, E)$ is a subset
$S \subseteq V$ of vertices, such that for every edge $e \in E$, at
least one of its end points is in $S$. The decision problem takes a
graph $G$ and a number $k' \ge 1$ as input and asks whether there is a
vertex cover of $G$ with size at most $k'$. Using the graph $G$, we
will construct a DFA $M_G$ over an alphabet $\Sigma_G$. We then show
that $G$ has a vertex cover of size $\le k'$ iff $M_G$ has an
equivalent DSA with total size $\le k$ where
$k = (k'+2)\times 2\theta + (2 \theta - 1)$. Here, $\theta$ is a
sufficiently large polynomial in $|V|, |E|$ which we will explain
later. 

  \begin{figure}
    \centering
   \scalebox{0.95} {\begin{tikzpicture}
      \begin{scope}
        \node (u) at (0,0) {\scriptsize $u$}; \node (v) at (-1.5, -1)
        {\scriptsize $v$}; \node (init) at (1, 1) {\scriptsize
          $q^{init}$}; \node (acc) at (1, -1) {\scriptsize $q^{acc}$};
        \node (sink) at (2.5, 0) {\scriptsize $q^{sink}$}; \node (v')
        at (-1.5, 1) {\scriptsize $v'$};
      \end{scope}
      \begin{scope}[->,>=stealth, thick, auto]
        \draw (v) to [bend left = 20] node {\scriptsize $e$} (u);
        \draw (u) to [bend left = 20] node {\scriptsize $e$} (v);
        \draw (init) to node [left, near start] {\scriptsize $u$} (u);
        \draw (u) to node {\scriptsize $e'', V, \theta$} (sink); \draw
        (u) to node [below] {\scriptsize $\$$} (acc);

        \draw (v') to [bend left=20] node {\scriptsize
          $e'$} (u); \draw (u) to [bend left=20] node [near end]
        {\scriptsize $e'$} (v');
      \end{scope}

      \begin{scope}[xshift=5.2cm]
        \begin{scope}
          \node (v) at (-1.5, -1) {\scriptsize $v$}; \node (init) at
          (1, 1) {\scriptsize $q^{init}$}; \node (acc) at (1, -1)
          {\scriptsize $q^{acc}$}; \node (sink) at (2.5, 0)
          {\scriptsize $q^{sink}$}; \node (v') at (-1.5, 1)
          {\scriptsize $v'$};
        \end{scope}
       
      \begin{scope}[->,>=stealth, thick, auto]
        \draw (init) to node [above] {\scriptsize $ue$} (v); \draw
        (init) to node [above] {\scriptsize $ue'$} (v'); \draw (init)
        to node {\scriptsize $u \$$} (acc); \draw (init) to node
        {\scriptsize $ue'', uV, u\theta$} (sink);
      \end{scope}
    \end{scope}

    \begin{scope}[xshift=11cm]
      \node (v) at (-1.5, -1) {\scriptsize $v$};
      \node (acc) at (1, -1) {\scriptsize
        $q^{acc}$}; \node (sink) at (2.5, 0) {\scriptsize
        $q^{sink}$}; \node (v') at (-1.5, 1) {\scriptsize $v'$};
    \end{scope}

      \begin{scope}[->,>=stealth, thick, auto]
        \draw (v) to [bend left = 20] node {\scriptsize $ee'$} (v');
        \draw (v') to [bend left=20] node {\scriptsize $e'e$} (v);
        \draw (v) to node [below] {\scriptsize $e\$$} (acc); \draw (v)
        to [bend left=20] node [near end] {\scriptsize $ee'', eV,
          e\theta$} (sink); \draw (v) to [loop left] node {\scriptsize
          $ee$} (v);
      \end{scope}
    \end{tikzpicture}}
    \caption{Left: Illustration of the neighbourhood of state
      $u$ in the DFA
      $M_G$. Middle, Right: Transitions induced from
      $q^{init}$ and $v$, on removing $u$.}
    \label{fig:complexity-mg}
  \end{figure}
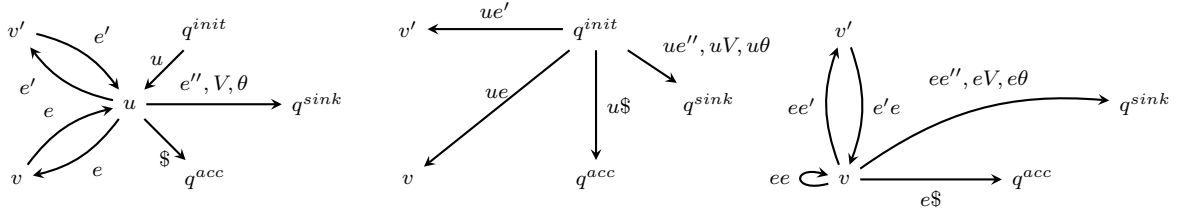

  The alphabet $\Sigma_G$ is given by $V \cup E \cup \{ \$\} \cup
  D$ where $D = \{1,2, \dots, \theta\}$. States of $M_G$ are $V \cup
  \{q^{init}, q^{sink},
  q^{acc}\}$. For simplicity, we use the same notation for
  $v$ as a vertex in $G$, $v$ as a letter in $\Sigma_G$ and
  $v$ as a state of $M_G$. The actual role of
  $v$ will be clear from the context. For every edge $e = (u,
  v)$, there are two transitions in the automaton: $u \xra{e}
  v$ and $v \xra{e} u$. For every $v \in
  V$, there are transitions $q^{init} \xra{v} v$ and $v \xra{\$}
  q^{acc}$. This automaton can be completed by adding all missing
  transitions to the sink state
  $q^{sink}$. Figure~\ref{fig:complexity-mg} (left) illustrates the
  neighbourhood of a state $u$. The notation
  $e''$ stands for any edge that is not incident on
  $u$; there is one transition for every such
  $e''$. Initial and accepting states are respectively
  $q^{init}$ and $q^{acc}$.
  Let
  $L_G(u)$ be the set of words that have an accepting run in
  $M_G$ starting from
  $u$ as the initial state. 
  If $u \neq v$, $L_G(u) = L_G(v)$ implies
  $(u,v)$ is an edge and there are no other edges outgoing either from
  $u$ or
  $v$. To avoid this corner case, we restrict the vertex cover problem
  to connected graphs of 3 or more vertices. Then we have
  $M_G$ to be a minimal DFA, with no two states equivalent. Here are
  two main ideas.

  \emph{Suppressing a state.} Suppose state $u$ of $M_G$
  is suppressed (i.e. $u$ is not in a suffix-tracking
  set). In Figure~\ref{fig:complexity-mg}, we show
  the induced transitions from $q^{init}$ and a vertex $v$. However,
  some of them will be redundant transitions: most importantly, the set
  of transitions $q^{init} \xra{u1, u2, \dots, u\theta} q^{sink}$ will
  be redundant bigger-suffix-transitions due to
  $q^{init} \xra{1, 2, \dots, \theta} q^{sink}$. Similarly,
  $v \xra{e1, e2, \dots, e\theta} q^{sink}$ will be removed. There are
  some more redundant bigger-suffix-transitions, like
  $v \xra{e e''} q^{sink}$ for some $e''$ that is not incident on $v$
  and $u$. So from each $v$, at most $2 |E|$ transitions are
  added. But crucially, after removing
  redundant transitions, the $\theta$ transitions from $u$
  no longer appear. If we choose $\theta$ large enough to compensate
  for the other transitions, we get an overall reduction in size by
  suppressing states.

  \emph{Two states connected by an edge cannot both be suppressed.}  Suppose $e = (u, v)$ is an edge. If $S$ is a set where
  $u, v \notin S$, then the transition $v \xra{e} u$ is not
  suffix-compatible: the simple word $ue$ from $q^{init}$ to $v$, when
  extended with $e$ gives the word $uee$; no suffix of $uee$ is a
  simple word from $q^{init}$ to $u$. We deduce that
  suffix-tracking sets in $M_G$ correspond to a vertex cover in $G$,
  and vice-versa.

  These two observations lead to a translation from minimum vertex
  cover to suffix-tracking sets with least number of states. Due to
  our choice of $\theta$, DSAs with smallest (total) size are indeed
  obtained from suffix-tracking sets with the least number of states. 
Let $k = (k' + 2) \times 2 \theta + (2 \theta - 1)$. Below, we elaborate these ideas in more detail and present the proof of the reduction.

\subparagraph*{Vertex cover $\le k'$ implies DSA $\le k$.}
Assume there is a vertex cover $\{v_1,
  \dots, v_p \}$ in $G$ with $p \le k'$. Let $S$ be the set of states in $M_G$ corresponding to $\{v_1,
  \dots, v_p \}$. Observe that $S \cup \{q^{init},q^{sink},q^{acc}\}$ is a suffix-tracking set; every transition is trivially suffix-compatible ($\forall q\xra{a}u, q\in S \text{ or } u\in S $). Well-formedness holds because $\forall p,q \in S, \a \in \spath{p}{q}{S}$ we have $|\a| \le 2$; this means $\forall q' \notin S, \beta \in \spath{p}{q'}{S}$, we have
  $\a \not \sfx \beta$ (since $|\beta|=1$). Hence the derived DSA will be equivalent to $M$.

  The derived DSA has $p + 3$ states, and transitions $q \xra{1, 2, \dots,
\theta} q^{sink}$ from each except for the $q^{sink}$ state. The
transitions on $q^{sink}$ are removable, and hence will be absent. All of this adds $(p+2) \times 2 \theta$ to the total size (edges +
label lengths). Apart from these, there are transitions with
labels of length at most $2$, over the alphabet $V \cup E \cup \$$. From each vertex, $v$, there are $|V|$ transitions to $q^{sink}$, one transition to $q^{acc}$ and at most $2|E|$ transitions to other states or $q^{sink}$. We
can choose a large enough $\theta$ (say $(|V| + |E|)^4$), so that the size of
these extra transitions is at most $2 \theta - 1$. 
  Hence, total size is $\le (p + 2) \times 2 \theta + (2 \theta - 1)$.

By assumption, we have $p \le k'$. Therefore, the size of the DSA 
is $\le (k' + 2) \times 2 \theta + (2 \theta - 1)
=k$.

  \subparagraph*{DSA $\le k$ implies vertex cover $\le k'$.}

  Let $\Aa$ be a DSA with size $\le k$. It may not be derived from $M_G$. However, by Proposition~\ref{prop:dsa-derivable-from-tracking-dfa} we know $\Aa$ is derived from a DFA $M$, the tracking DFA for $\Aa$. Moreover since $M_G$ is the minimal DFA, we know that $M$ will be a \emph{refinement} of $M_G$ (see Section~\ref{sec:preliminaries} for definition).

  Let us consider a pair of states $u$ and $v$ from $M_G$, such that the vertices $u,v \in G$ have an edge between them labeled $e$. The DFA $M$ will have two sets of states $u_1,u_2,\dots,u_i$ and $v_1,v_2,\dots,v_j$ that are language-equivalent to $u$ and $v$ respectively. Its initial state must have a transition on $v$ to one of $v_1,v_2,\dots,v_j$. Without loss of generality, let it be to $v_1$. Each of $v_1,v_2,\dots,v_j$ must have a transition on $e$ to one of $u_1,u_2,\dots,u_i$ (for equivalence with $M_G$) and vice-versa. Consider the run from the initial state on $ve^{i+j+1}$. At least one of the states among $u_1,u_2,\dots,u_i, v_1,v_2,\dots,v_j$ must be visited twice; consider the first such instance. The transition on $e$ that re-visits a state cannot be suffix-compatible w.r.t a set $S$, if none of these states are in $S$. For it to be suffix-compatible, the string $ve^k.e$ (from initial state to the first repeated state) must have its longest simple-word suffix go to the same state. Since $ve^k.e$ is not simple by itself, its longest suffix must consist entirely of $e$'s. But on any string of $e$'s, the initial state moves only to the sink state(s) and not to any of $u_1,u_2,\dots,u_i, v_1,v_2,\dots,v_j$. Hence any suffix-tracking set must contain at least one of these states, which maps to at least one of $u$ or $v$ in $G$. Every suffix-tracking set of $M$ therefore maps to a vertex cover $\{v_1, v_2, \dots, v_p\}$.

  Now we show that the size of this vertex cover is $\le k'$.
Each of the states picked in the suffix-tracking set will contribute to atleast $2 \theta$ in the total size, due to the $\theta$ transitions. We will also have these $\theta$ transitions from the initial and accepting states. Therefore, the total size is $(p + 2) \times 2 \theta + y$ for some $y > 0$. Hence $(p + 2) \times 2 \theta \le k$. This implies $p \le k'$: otherwise we will have $p \ge k'+1$, and hence $(p + 2) \times 2 \theta \ge (k' + 1 + 2) \times 2 \theta = (k' + 2) \times 2 \theta + 2 \theta > k$, a contradiction.


\section{Conclusion}
\label{sec:conclusion}

We have introduced the model of deterministic suffix-reading automata,
compared its size with DFAs and DGAs, proposed a method to derive DSAs
from DFAs, and presented the complexity of minimization. 
The work on DGAs~\cite{giammarresi1999deterministic} inspired us to look for methods to derive DSAs from
DFAs, and investigate whether they lead to minimal DSAs for a
language. This led to our technique of suffix-tracking sets, which
derives DSAs from DFAs. The technique imposes some natural conditions
on subsets of states, for them to be tracking patterns at each
state.
However, surprisingly, the smallest DSA that we can derive from the
canonical DFA need not correspond to the minimal DSA of a
language. We have shown that when restricting the syntax, our derivation method is able to generate a minimal DSA in the restricted class, starting from the canonical DFA. 

In this introductory work on DSAs, our goal has been to present the model, its motivations and establish ingredients for a deeper study of the model, both from a practical and a theoretical perspective. Recently, we have enhanced the DSA syntax to include a parallel composition operator $\parallel$ in the transitions~\cite{netys}. This helps succinct representation of the patterns when the alphabet is distributed across multiple components in a concurrent system. Using the enhanced model, we have provided a formal semantics and test generation algorithm (with guarantees) for an industrial formalism called Expressive Decision Tables (EDT)~\cite{DBLP:conf/date/VenkateshSKA14} developed by the industry partners in this work. This work shows a direct impact of the DSA model in an industrial setting.

From a theoretical perspective, there are plenty of problems to ponder about. We do not yet have an algorithm that can start with the canonical DFA, perform some operations on it and get a minimal DSA (in the general case, and not the strong DSA case as discussed in Section~\ref{sec:strong}). As we saw in Section~\ref{sec:minim-some-observ}, one might need to expand the canonical DFA to get the minimal DSA. It would be interesting to have a clear algorithm to identify this expansion on which the suffix-tracking techniques can be applied. 
Can we use our techniques to study minimality in terms of number of
states?  Closure properties of DSAs - can we perform the union,
intersection and complementation operations on DSAs without computing
the entire equivalent DFAs? What about Myhill-Nerode style congruences for DSAs? Recently a Myhill-Nerode theorem has been established for deterministic generalized automata~\cite{DBLP:conf/stacs/Cotumaccio24}.  To sum up, we believe the DSA model offers advantages in the
specification of systems and in also studying regular languages
from a different angle. The results that we have presented throw
light on some of the different aspects in this model, and lead to
many questions both from theoretical and practical
perspectives.


\bibliographystyle{alphaurl}
\bibliography{dsa}

\end{document}